\documentclass[aps,pra,10pt,twocolumn,floatfix,superscriptaddress,showkeys]{revtex4-2}
\usepackage{graphicx}
\usepackage{amsmath,amssymb,amsthm,mathtools,bm}
\usepackage{booktabs}
\usepackage{array}
\usepackage{tabularx}
\usepackage{multirow}
\usepackage{xcolor}
\usepackage{microtype}
\usepackage{enumitem}
\usepackage{placeins}
\usepackage{physics}
\usepackage{dsfont}
\usepackage[colorlinks=true,citecolor=blue,linkcolor=blue,urlcolor=blue]{hyperref}
\usepackage[capitalize]{cleveref}
\graphicspath{
{sp_hqaoa_paper_sim-v6/figures/main/}
{sp_hqaoa_paper_sim-v6/figures/appendix/}
}
\newtheoremstyle{apsthm}%
{\topsep}
{\topsep}
{\normalfont}
{\parindent}
{\itshape}
{.}
{ }
{\thmname{#1}\thmnumber{ #2}\thmnote{ (#3)}}

\theoremstyle{apsthm}
\newtheorem{theorem}{Theorem}
\newtheorem{proposition}{Proposition}
\newtheorem{lemma}{Lemma}
\newtheorem{corollary}{Corollary}

\newcommand{\wtot}{W_{\rm tot}}


\begin{document}

\title{Symmetry-Preserving Hierarchical Quantum Approximate Optimization Algorithm with Subproblems Beyond the Qubit Limit}

\author{Zhaoyang Song}
\affiliation{Center for Quantum Sciences and School of Physics, Northeast Normal University, Changchun 130024, China}
\author{Wei Zhang}
\affiliation{School of Physics, Harbin Institute of Technology, Harbin, Heilongjiang 150001, China}
\author{Ji Cao}
\email{caoji@ybu.edu.cn}
\affiliation{Department of Physics, College of Science, Yanbian University, Yanji, Jilin 133002, China}
\affiliation{Institute of Quantum Science and Technology, Yanbian University, Yanji, Jilin 133002, China}
\author{Xuexi Yi}
\email{yixx@nenu.edu.cn}
\affiliation{Center for Quantum Sciences and School of Physics, Northeast Normal University, Changchun 130024, China}
\author{Hong-Fu Wang}
\email{hfwang@ybu.edu.cn}
\affiliation{Center for Quantum Sciences and School of Physics, Northeast Normal University, Changchun 130024, China}
\affiliation{Department of Physics, College of Science, Yanbian University, Yanji, Jilin 133002, China}
\affiliation{Institute of Quantum Science and Technology, Yanbian University, Yanji, Jilin 133002, China}
\affiliation{Institute for Interdisciplinary Quantum Information Technology, Jilin Engineering Normal University, Changchun, Jilin 130052, China}
\affiliation{Jilin Engineering Laboratory for Quantum Information Technology, Changchun, Jilin 130052, China}

\date{\today}

\begin{abstract}
The quantum approximate optimization algorithm (QAOA) is a paradigmatic variational quantum algorithm for combinatorial optimization problems, yet on current quantum devices the available qubit count and achievable circuit depth remain the main bottlenecks to its practical application. Divide-and-merge reduces the number of variables encoded in a single quantum circuit, while contraction based on coarse graining extends the effective range of shallow circuits, but existing work has not addressed both challenges simultaneously. We propose symmetry-preserving hierarchical QAOA (SP-HQAOA), in which parity sector contraction preserves the zero-field Ising structure of the effective objective at every contraction step, thereby addressing both within a single algorithm. We further prove that this reduction is the unique one with this property. Successive contractions further decouple the maximum leaf block size $B$ from the qubit limit $Q$ per quantum call, so every quantum circuit uses at most $Q$ qubits even when $B>Q$. We also establish a deterministic repair-risk bound on the possible loss in optimal value caused by each sector restriction, and the resulting bounds can be summed across contraction steps and levels of the recursive hierarchy. With depth $p=1$, SP-HQAOA attains $99.4\%$--$99.7\%$ of the exact optimum on $N=40$ random regular MaxCut leaf blocks and, with at most $10$ qubits per circuit, $96\%$ of the classical reference for $N=500$. SP-HQAOA demonstrates the capacity to use shallow quantum circuits to address combinatorial optimization problems whose variable count far exceeds the qubit limit of a single quantum call, and may enable practical quantum optimization advantages in the near term.
\end{abstract}

\maketitle

\section{Introduction}

The quantum approximate optimization algorithm (QAOA) is a canonical variational algorithm for Ising formulations of combinatorial optimization~\cite{barahona1982complexity,lucas2014ising,mohseni2022ising} on noisy intermediate-scale quantum (NISQ) devices~\cite{preskill2018nisq,bharti2022nisq,farhi2014qaoa,zhou2020qaoa,blekos2024review}. It addresses problems such as MaxCut, spin-glass optimization, and quadratic unconstrained binary optimization~\cite{goemans1995improved,wang2018qaoa,wurtz2021maxcut,farhi2022sk,boros2002pseudo,ebadi2022mis,kim2022rydberg,bravyi2022coloring,brandao2022sdp,boulebnane2024sat,shaydulin2024scaling}, where spin configurations with lower energy correspond to better candidate solutions. QAOA alternates problem and mixing evolutions to prepare a parameterized quantum state, whose parameters are optimized variationally~\cite{harrigan2021qaoa,pagano2020qaoa}. This architecture transforms a discrete optimization objective into a quantum state preparation problem. However, the practical applicability of QAOA remains constrained by the available number of qubits and the implementable circuit depth~\cite{guerreschi2019qaoa,farhi2020typical,farhi2020worstcase}. A problem with $N$ variables typically requires $O(N)$ qubits, so the available quantum register bounds the encodable problem size. Due to the constraint imposed by the quantum circuit light cone on information propagation~\cite{lieb1972finite,bravyi2006lieb}, QAOA at fixed depth $p$ can establish correlations only within a finite graph distance~\cite{hastings2019local,akshay2020reachability}. Solving these two fundamental problems is a prerequisite for QAOA to demonstrate practical utility on NISQ devices.

Graph partitioning, divide-and-merge, recursive decomposition, and multilevel quantum local search decompose a large instance into smaller subproblems, solve them independently, and reconstruct a global solution by merging or successive refinement~\cite{zhou2023qaoa2,li2021dcqaoa,ponce2025graph,ushijima2021multilevel}. These methods reduce the number of variables per quantum call, but the correlations that can be established within each subproblem remain bounded by the light cone of a shallow circuit. Recursive QAOA and multiscale QAOA based on real-space renormalization instead use local measurement outcomes to select variable constraints or effective subspaces and propagate this information through successively reduced problems, allowing shallow circuits to act on increasingly large effective scales~\cite{bravyi2020obstacles,zou2025mqaoa}. However, these methods provide no mechanism for imposing a fixed upper limit on the number of variables in each quantum call, so the largest single call is not decoupled from the size of the current problem. A natural question is whether the two approaches can be combined to reduce the number of variables per quantum call while extending the effective range of shallow circuits.

Combining the two requires first resolving a structural inconsistency that arises under iteration. Exact block flip merge can be applied recursively because each merge produces another zero-field quadratic optimization problem of the same class~\cite{zhou2023qaoa2}. A generic real-space coarse graining can compress local degrees of freedom but may generate one-body field terms in the effective objective, taking the reduced problem out of the zero-field Ising class and preventing the same contraction and merge rules from being reused at the next level. Therefore, the combination requires a renormalization scheme that preserves the zero-field Ising structure throughout successive contractions.

We propose symmetry-preserving hierarchical QAOA (SP-HQAOA, Fig.~\ref{fig:incompatibility}), which resolves this structural inconsistency. At each contraction step, SP-HQAOA fixes the relative parity of the endpoints of a matched edge and uses one effective spin to represent the two retained configurations. Under this substitution, the internal edge reduces to a constant, while the boundary couplings remain two-body products. Therefore, the reduction generates no one-body fields, and the effective objective remains within the zero-field signed Ising class. This zero-field closure ensures that parity sector contraction and exact block flip merging remain applicable at subsequent recursive levels. We further prove that fixing a parity sector is the unique choice compatible with zero-field closure. For each selected contraction edge, the retained sector is determined from the parity statistic \(\langle Z_iZ_j\rangle\) obtained by shallow QAOA from the same computational-basis samples used for energy estimation. When the maximum leaf block size \(B\) exceeds the qubit limit \(Q\) of a single quantum call, SP-HQAOA progressively reduces the number of effective variables through successive contractions while keeping every QAOA call to at most \(Q\) qubits. Because successive contractions require multiple parity sectors to be fixed, each sector restriction may reduce the optimal value. To control this loss, we establish a deterministic repair bound \(R_{ij}(\eta)\). These bounds can be summed within a matching and accumulated across contraction steps and recursive levels.

MaxCut is the canonical combinatorial optimization problem studied with QAOA. Our numerical simulations on MaxCut instances verify the contraction identity and repair-risk bound, and show that parity statistics obtained from shallow QAOA effectively guide sector selection. On \(40\)-vertex random regular leaf blocks, SP-HQAOA reaches \(99.4\%\)--\(99.7\%\) of the exact optimum using depth-\(p=1\) circuits followed by classical refinement. On weighted \(3\)-regular instances, fixing \(Q=10\) and progressively increasing \(B\) reduces the fraction of cross-block edge weight and steadily improves the cut value. At \(N=500\) and \(B=80\), SP-HQAOA still reaches \(96\%\) of the classical reference value. This framework allows local correlations extracted by shallow QAOA to contribute to the solution of larger problems through successive contractions and recursive merging, opening new possibilities for large-scale combinatorial optimization with limited quantum resources.

\section{Problem setting and structural obstruction}
\label{sec:setting-obstruction}

This section states the algebraic conditions required below. We first define zero-field signed Ising objectives and extend the exact block flip merge to this class, then show how one-body fields can arise under coarse-graining contractions and alter the form of the effective objective.

\subsection{Zero-field signed Ising objectives}
\label{sec:problem}

Let $G=(V,E)$ be a weighted undirected graph with edge weights $W_{ij}\ge 0$. Given a binary assignment $z\in\{\pm1\}^V$, the MaxCut problem~\cite{goemans1995improved} maximizes the cut value
\begin{equation}
C_G(z)=\frac{1}{2}\sum_{(i,j)\in E}W_{ij}(1-z_i z_j).
\label{eq:maxcut}
\end{equation}
An edge contributes to the cut exactly when its endpoints have opposite signs. Maximizing $C_G$ is equivalent to minimizing the zero-field Ising energy
\begin{equation}
H_G(z)=\sum_{(i,j)\in E}W_{ij}z_i z_j,
\label{eq:maxcut-ising}
\end{equation}
since $C_G(z)=\tfrac{1}{2}\sum_{(i,j)\in E}W_{ij}-\tfrac{1}{2}H_G(z)$. Equation~\eqref{eq:maxcut-ising} contains no one-body terms proportional to a single spin \(z_i\), a property we call the zero-field structure.

The MaxCut problem corresponds to the special case in which all couplings are nonnegative. Allowing arbitrary real couplings gives the zero-field signed Ising objective~\cite{barahona1982complexity}
\begin{equation}
H(z)=\kappa+\sum_{\{i,j\}\subseteq V}J_{ij}z_i z_j,
\qquad J_{ij}\in\mathbb{R},
\label{eq:signed-ising}
\end{equation}
where uncoupled pairs have \(J_{ij}=0\), and $\kappa$ is independent of $z$. We retain this constant because successive reductions naturally generate additive constants. We use the symmetric convention $J_{ab}=J_{ba}$ and call a pair of spins with $J_{ij}\ne0$ an edge of $H$.

In Eq.~\eqref{eq:signed-ising}, each nonzero coupling $J_{ij}$ determines a particular relative parity between $z_i$ and $z_j$. Minimizing $J_{ij}z_i z_j$ gives $z_i z_j=-\operatorname{sgn}(J_{ij})$. More explicitly, each coupling term admits the exact decomposition
\begin{equation}
J_{ij}z_i z_j = -|J_{ij}| + 2|J_{ij}|\,\mathbf{1}[z_i z_j \ne -\operatorname{sgn}(J_{ij})],
\label{eq:signed-term-parity-indicator}
\end{equation}
where the second term is a nonnegative weight times a parity mismatch indicator.

\subsection{Exact block flip merge}
\label{sec:merge}

Suppose the variable set $V$ is partitioned into blocks $V_1,\ldots,V_h$. For each block $V_r$, let $x^{(r)}\in\{\pm1\}^{V_r}$ be a fixed assignment returned by a local solver. Replacing $x^{(r)}$ with $-x^{(r)}$ leaves every product $x_u^{(r)}x_v^{(r)}$ with $u,v\in V_r$ unchanged. Only terms coupling different blocks can distinguish the two assignments. To encode this choice, we introduce a block flip variable $\sigma_r\in\{\pm1\}$ such that
\begin{equation}
z_u = \sigma_r\,x_u^{(r)}, \qquad u\in V_r.
\end{equation}
Choosing $\boldsymbol{\sigma}=(\sigma_1,\ldots,\sigma_h)$ is the merge problem, illustrated in Fig.~\ref{fig:incompatibility}(a). The following proposition is a direct algebraic extension of the QAOA\(^{2}\) merge theorem~\cite{zhou2023qaoa2} from MaxCut with nonnegative couplings to the zero-field signed Ising class.

\begin{proposition}[Exact signed Ising merge]
\label{prop:merge}
For fixed local assignments $x^{(r)}$ in a given zero-field signed Ising problem, optimization over the block flip variables $\boldsymbol{\sigma}=(\sigma_1,\ldots,\sigma_h)$ is equivalent to minimizing another zero-field signed Ising objective
\begin{equation}
H_{\rm merge}(\boldsymbol{\sigma})=\kappa_{\rm merge}+\sum_{1\le r<t\le h}K_{rt}\,\sigma_r\sigma_t,
\label{eq:merge-ising}
\end{equation}
where $\kappa_{\rm merge}$ collects all terms independent of $\boldsymbol{\sigma}$, and the effective couplings are
\begin{equation}
K_{rt}=\sum_{u\in V_r}\sum_{v\in V_t}J_{uv}\,x_u^{(r)}x_v^{(t)}.
\label{eq:merge-coupling}
\end{equation}
All terms involving only variables in a single block are independent of $\boldsymbol{\sigma}$.
\end{proposition}

\begin{proof}
Substitute $z_u=\sigma_r x_u^{(r)}$ into Eq.~\eqref{eq:signed-ising}. For $u,v\in V_r$, the corresponding coupling term becomes
\begin{equation}
J_{uv}\,\sigma_r x_u^{(r)}\,\sigma_r x_v^{(r)}=J_{uv}\,x_u^{(r)}x_v^{(r)},
\end{equation}
because $\sigma_r^2=1$. Hence all terms involving only variables in $V_r$ are independent of $\boldsymbol{\sigma}$. For $u\in V_r$ and $v\in V_t$ with $r\ne t$, the term coupling $V_r$ and $V_t$ becomes
\begin{equation}
J_{uv}\,x_u^{(r)}x_v^{(t)}\,\sigma_r\sigma_t.
\end{equation}
Collecting all terms that couple distinct blocks gives Eq.~\eqref{eq:merge-coupling}, while all terms involving only variables in a single block and the original constant are absorbed into $\kappa_{\rm merge}$.
\end{proof}

By Proposition~\ref{prop:merge}, the merge objective of a zero-field signed Ising problem is again an objective in the same class, so whenever the merge graph still exceeds the available qubit count, it can itself be partitioned, solved blockwise, and merged again.

This closure is precisely the structural basis exploited by QAOA\(^{2}\) for MaxCut~\cite{zhou2023qaoa2}, where invariance of each block objective under the global flip allows the merge problem to be written as another MaxCut instance.

\subsection{Coarse graining and zero-field closure}
\label{sec:obstruction}

Coarse-graining contractions need not preserve the zero-field structure. In a deterministic reduction of two spins to one effective spin \(s\in\{\pm1\}\), if \(J_{ij}\ne0\) and the two retained assignments of \((z_i,z_j)\) differ at only one endpoint, the internal coupling \(J_{ij}z_iz_j\) reduces to \(\pm J_{ij}s\), producing a one-body field.

MQAOA provides a concrete example in which a quantum subspace projection loses zero-field closure. It constructs a two-dimensional effective subspace for a pair of spins from the two-qubit reduced density matrix of the QAOA state~\cite{zou2025mqaoa}. The retained basis has two product forms, one of which can be written as
\begin{equation}
|0_s\rangle = |0_i\rangle|\psi_0\rangle_j,
\qquad
|1_s\rangle = |1_i\rangle|\psi_1\rangle_j,
\label{eq:mqaoa-basis}
\end{equation}
where $|\psi_0\rangle$ and $|\psi_1\rangle$ are conditional states of qubit $j$ determined by the reduced density matrix. Let $\omega:\mathcal{H}_s\to\mathcal{H}_i\otimes\mathcal{H}_j$ denote the isometry whose columns are the two retained basis states. In the basis of Eq.~\eqref{eq:mqaoa-basis},
\begin{equation}
\omega^\dagger Z_i \omega = Z_s,
\end{equation}
whereas $Z_j$ generally maps to
\begin{equation}
\omega^\dagger Z_j \omega = \mu_I I+\mu_Z Z_s,
\label{eq:generic-zj-map}
\end{equation}
with
\begin{equation}
\begin{aligned}
\mu_I&=\frac{1}{2}\bigl(\langle\psi_0|Z|\psi_0\rangle+\langle\psi_1|Z|\psi_1\rangle\bigr),\\
\mu_Z&=\frac{1}{2}\bigl(\langle\psi_0|Z|\psi_0\rangle-\langle\psi_1|Z|\psi_1\rangle\bigr).
\end{aligned}
\label{eq:mu-def}
\end{equation}
The coefficient $\mu_I$ is generically nonzero unless the two conditional states have exactly opposite $Z$ polarizations. The other basis form used in MQAOA exchanges the roles of $i$ and $j$, so the identity component appears in the effective representation of $Z_i$ rather than $Z_j$.

If $\mu_I\ne0$, projecting a zero-field two-body boundary contribution can generate a one-body term. Suppose $i$ and $j$ are both coupled to the same external vertex $k\notin\{i,j\}$. Omitting identity factors on unaffected sites,
\begin{equation}
\omega^\dagger
\left(J_{ik}Z_iZ_k+J_{jk}Z_jZ_k\right)
\omega
=
\mu_I J_{jk}Z_k+
\left(J_{ik}+\mu_Z J_{jk}\right)Z_sZ_k.
\label{eq:combined-boundary-field-leakage}
\end{equation}
The term $\mu_I J_{jk}Z_k$ is a one-body field on $k$. The same identity component also produces a field through the internal edge as
\begin{equation}
\omega^\dagger(J_{ij}Z_iZ_j)\omega
=
J_{ij}\bigl(\mu_I Z_s+\mu_Z I\bigr).
\label{eq:internal-field-leakage}
\end{equation}
Thus, the projected Hamiltonian remains an Ising Hamiltonian but need not remain zero-field.

For the projection onto the basis in Eq.~\eqref{eq:mqaoa-basis}, \(\mu_I=0\) removes the field terms shown in Fig.~\ref{fig:incompatibility}(b). A deterministic lift also requires \(|\mu_Z|=1\), because when \(|\mu_Z|<1\), the retained states do not assign a definite \(Z\) value to the second spin.

\begin{figure*}[t]
\centering
\includegraphics[width=0.94\textwidth]{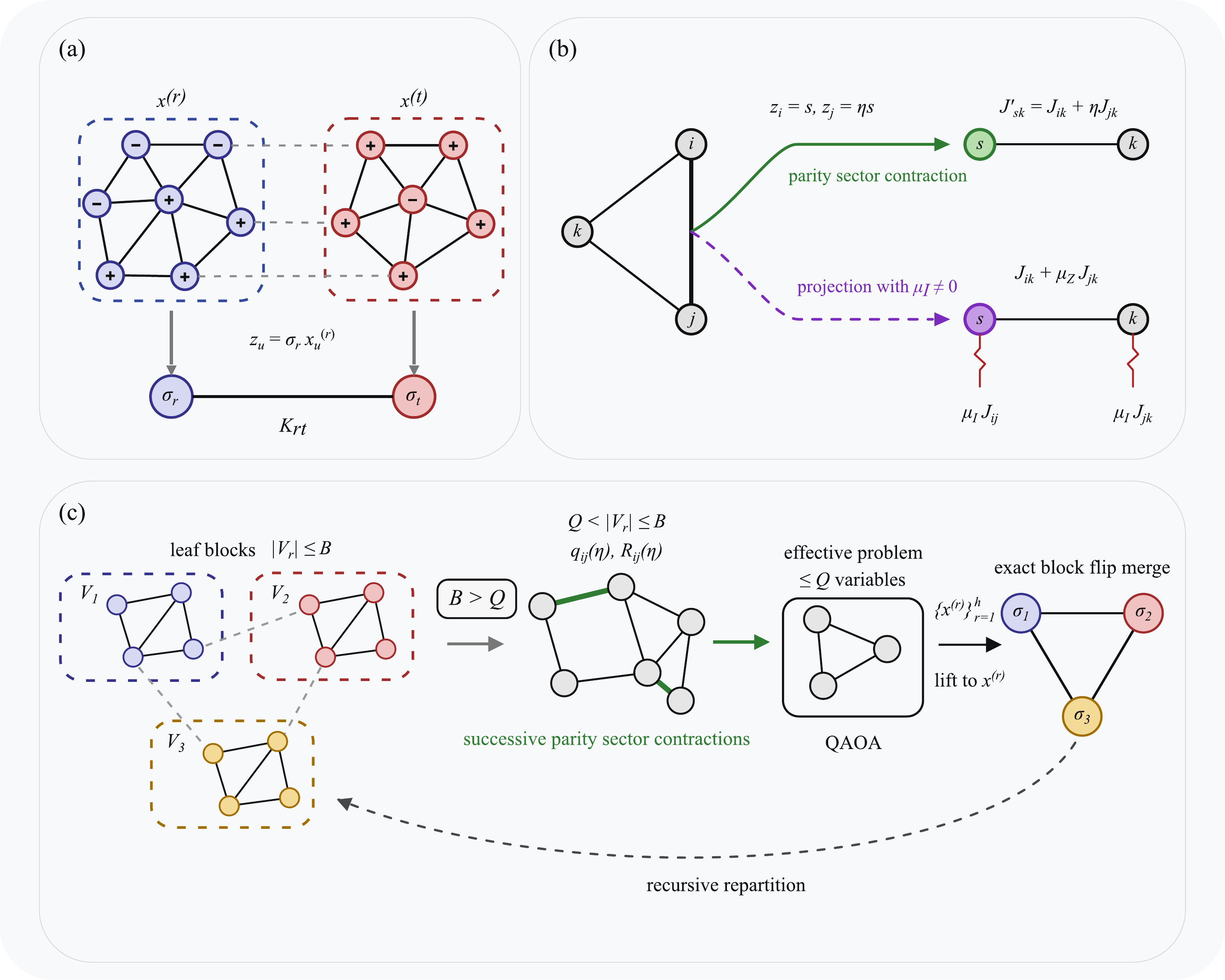}
\caption{Block flip merge, parity sector contraction, and the \(B>Q\) hierarchy.
(a)~Exact block flip merge. Global flips of fixed assignments \(x^{(r)}\) and \(x^{(t)}\) are parameterized by \(z_u=\sigma_r x_u^{(r)}\), producing a zero-field signed Ising merge problem with effective coupling \(K_{rt}\).
(b)~Two treatments of the same pair \((i,j)\). Parity sector contraction uses \(z_i=s\) and \(z_j=\eta s\), giving the two-body coupling \(J'_{sk}=J_{ik}+\eta J_{jk}\) without one-body fields. For a projection with \(\mu_I\ne0\), up to an additive constant, the two-body coefficient becomes \(J_{ik}+\mu_ZJ_{jk}\), with one-body coefficients \(\mu_IJ_{ij}\) and \(\mu_IJ_{jk}\) on \(s\) and \(k\), respectively.
(c)~Outer leaf blocks satisfy \(|V_r|\le B\), while each QAOA problem contains at most \(Q\) variables and \(B\) may exceed \(Q\). Blocks with at most \(Q\) variables are solved directly. A representative oversized block \(Q<|V_r|\le B\) is reduced by repeated matched contractions. The QAOA sector probabilities \(q_{ij}(\eta)\) and repair risks \(R_{ij}(\eta)\) jointly determine sector selection subject to \(\sum_{(i,j)\in M}R_{ij}(\eta_{ij})\le T_R\). Each reduced block is solved and lifted to \(x^{(r)}\). All leaf assignments enter the exact block flip merge of panel~(a), and a merge problem exceeding \(Q\) variables is recursively repartitioned.}
\label{fig:incompatibility}
\end{figure*}

\section{Symmetry-preserving parity renormalization}
\label{sec:parity-rg}

We first prove that the reduction obtained by imposing the relative parity constraint \(z_iz_j=\eta\) on an edge keeps the objective within the zero-field signed Ising class and that it is the only deterministic single-edge reduction with this closure property. We then extend the reduction to a matching, iterate it across contraction steps, derive sector probabilities from QAOA parity statistics, and show how iteration extends the effective range of a shallow circuit.

\subsection{Parity sector contraction of a single edge}
\label{sec:single-edge-contraction}

Consider an edge \((i,j)\) of a zero-field signed Ising objective \(H\). We choose \(\eta\in\{\pm1\}\) and impose the relative parity constraint \(z_iz_j=\eta\). The pair then has one degree of freedom, parameterized by \(s\). We define the lift map \(\Lambda_\eta\), which maps \(s\) and the remaining spins \(z_{\rm rest}\) to an assignment of \(H\), by
\begin{equation}
\begin{aligned}
(\Lambda_\eta(s,z_{\rm rest}))_i&=s,\\
(\Lambda_\eta(s,z_{\rm rest}))_j&=\eta s,\\
(\Lambda_\eta(s,z_{\rm rest}))_k&=z_k,
\qquad k\notin\{i,j\}.
\end{aligned}
\label{eq:parity-embedding}
\end{equation}
By Eq.~\eqref{eq:parity-embedding} and \(s^2=1\), every assignment produced by \(\Lambda_\eta\) satisfies \(z_iz_j=\eta\). Conversely, any assignment \(z\) satisfying this constraint has \(z_j=\eta z_i\). Setting \(s=z_i\) and taking the remaining components as \(z_{\rm rest}\) gives \(\Lambda_\eta(s,z_{\rm rest})=z\). Therefore, \(\operatorname{im}(\Lambda_\eta)\) consists exactly of the assignments satisfying this constraint, and this set is the parity sector labeled by \(\eta\). Composing \(H\) with \(\Lambda_\eta\) gives the reduced objective \(H(\Lambda_\eta(s,z_{\rm rest}))\) in \(s\) and \(z_{\rm rest}\). We call this reduction parity sector contraction. Notably, the same substitution rule also appears in the recursive QAOA of Ref.~\cite{bravyi2020obstacles} as a variable elimination step. The difference is that Ref.~\cite{bravyi2020obstacles} selects an edge with the largest correlation magnitude, takes \(\eta\) to be the sign of that correlation, and eliminates one variable per recursion step, whereas here the rule is applied to all edges of a matching, contracting several edges in one step.

\begin{proposition}[Parity sector contraction identity]
\label{prop:sector-exact}
Let $H(z)=\kappa+\sum_{a<b}J_{ab}\,z_a z_b$ be a zero-field signed Ising objective.
Contracting edge $(i,j)$ in sector $\eta$ gives
\begin{equation}
\begin{aligned}
	H(\Lambda_\eta(s,z_{\rm rest}))
	&=
	\kappa'
	+\sum_{k\notin\{i,j\}}
	\bigl(J_{ik}+\eta J_{jk}\bigr)\,s\,z_k\\
	&\quad
	+\sum_{\substack{k<\ell\\k,\ell\notin\{i,j\}}}
	J_{k\ell}\,z_kz_\ell,
\end{aligned}
\label{eq:sector-exact}
\end{equation}
where $\kappa'=\kappa+J_{ij}\eta$.
The contracted objective remains in the zero-field signed Ising class.
\end{proposition}

\begin{proof}
Substituting $z_i=s$ and $z_j=\eta s$ into $H(z)$ produces three groups of terms.
The internal edge becomes $J_{ij}\,s\,\eta s=J_{ij}\eta$, which is independent of $s$ and absorbed into $\kappa'$.
For each external vertex $k\notin\{i,j\}$, the two incident terms combine as $J_{ik}\,s\,z_k+J_{jk}\,\eta s\,z_k=J'_{sk}\,s\,z_k$, where $J'_{sk}=J_{ik}+\eta J_{jk}$.
Because the combined term is a two-body product, no one-body field is generated.
Terms not incident to $i$ or $j$ are unaffected.
Therefore, every term in the contracted objective is either a constant or a two-body product, confirming zero-field signed Ising closure.
\end{proof}

Take the columns of \(\omega\) to be the two computational basis states retained in the sector. Then \(z_i=s\) and \(z_j=\eta s\) give \(\omega^\dagger Z_i\omega=Z_s\) and \(\omega^\dagger Z_j\omega=\eta Z_s\). Therefore, this representation has \(\mu_I=0\) and \(\mu_Z=\eta\), hence \(|\mu_Z|=1\). These are the two conditions identified at the end of Sec.~\ref{sec:obstruction}. Figure~\ref{fig:incompatibility}(b) contrasts parity sector contraction with a projection for which \(\mu_I\ne0\).

\begin{proposition}[Uniqueness of deterministic zero-field reduction of a single edge]
\label{prop:pair-uniqueness}
Let \(H\) be a zero-field signed Ising objective, and let \((i,j)\) be an edge with \(J_{ij}\ne0\). Consider any reduction that maps the two values of \(s\) to two distinct assignments of \((z_i,z_j)\) and leaves all other spins unchanged. Suppose that substituting these assignments into \(H\) yields a zero-field signed Ising objective in \(s\) and the remaining spins. Up to the relabeling \(s\mapsto-s\), every such reduction has the form \(z_i=s\) and \(z_j=\eta s\), where \(\eta\in\{\pm1\}\).
\end{proposition}

\begin{proof}
Because the two retained assignments are distinct, at least one of \(z_i\) and \(z_j\) changes sign between \(s=+1\) and \(s=-1\). Suppose that exactly one endpoint changes sign. The internal term \(J_{ij}z_iz_j\) then reduces to \(\pm J_{ij}\,s\). Every boundary term becomes a multiple of \(s z_k\) or \(z_k\). Terms not incident to \(i\) or \(j\) are unchanged. Hence no other term contributes a monomial proportional to \(s\) alone. Therefore, the reduced objective contains the one-body field \(\pm J_{ij}\,s\), which contradicts zero-field closure because \(J_{ij}\ne0\). Hence both endpoints change sign. The two retained assignments of \((z_i,z_j)\) are then global flips of one another. Relabeling \(s\) if necessary gives \(z_i=s\) and \(z_j=\eta s\).
\end{proof}

The zero-field structure is equivalent to invariance of \(H\) under the global flip \(z\mapsto-z\). This flip splits the four configurations of \((z_i,z_j)\) into two orbits, and the two parity sectors each cover one orbit on this pair. Equation~\eqref{eq:parity-embedding} gives \(\Lambda_\eta(-s,-z_{\rm rest})=-\Lambda_\eta(s,z_{\rm rest})\), that is, the global flip of the original variables corresponds to the global flip of the reduced variables. Therefore, parity sector contraction preserves the symmetry on which Proposition~\ref{prop:merge} relies. By Proposition~\ref{prop:sector-exact}, the reduced objective remains in the zero-field signed Ising class, and the block flip merge applies to it unchanged.

\subsection{Matched contractions and iterated closure}
\label{sec:matched-closure}

Set $H^{(0)}=H$. For each $\ell$, let $V^{(\ell)}$ denote the variable set of $H^{(\ell)}$, call its elements effective vertices, and write $s_A$ for the spin at $A\in V^{(\ell)}$, with $s_u=z_u$ for $u\in V^{(0)}$. The edges of $H^{(\ell)}$ form the effective graph $G^{(\ell)}$ on $V^{(\ell)}$. To obtain $H^{(\ell+1)}$ from $H^{(\ell)}$, we choose a matching $M$ in $G^{(\ell)}$, select a parity $\eta_e$ for each $e\in M$, and contract all edges of $M$ simultaneously. Effective vertices not covered by $M$ are carried unchanged into $V^{(\ell+1)}$.

\begin{proposition}[Closure under matched parity sector contraction]
\label{prop:matched-closure}
Let $H$ be a zero-field signed Ising objective and $M$ any matching consisting of edges of $H$. For any assignment of parity sectors to the edges of $M$, the objective obtained by contracting all edges of $M$ simultaneously remains in the zero-field signed Ising class.
\end{proposition}

\begin{proof}
The edges of $M$ share no endpoints, so their single-edge contractions act on disjoint pairs of spins. By Eq.~\eqref{eq:sector-exact}, contracting one edge leaves every coupling not incident to its endpoints unchanged, so the other edges of $M$ remain edges after that step. Applying these substitutions one at a time gives the same reduced objective as applying them simultaneously. By Proposition~\ref{prop:sector-exact}, every single-edge contraction in the sequential construction leaves the objective in the zero-field signed Ising class, and so does the simultaneous contraction.
\end{proof}

By Proposition~\ref{prop:matched-closure}, one matched contraction of a zero-field signed Ising objective produces another objective in the same class, so matched contractions can be applied repeatedly. Let $M_\ell$ be the matching used at step $\ell$, let $\Lambda_\ell$ be the combination of the lift maps in Eq.~\eqref{eq:parity-embedding} over the edges of $M_\ell$, and let $\kappa_\ell$ be the sum of the new additive constants that these edges produce in Eq.~\eqref{eq:sector-exact}. We define the objective after this step by
\begin{equation}
H^{(\ell+1)}
=
H^{(\ell)}\circ\Lambda_\ell-\kappa_\ell.
\label{eq:level-recursion}
\end{equation}
The composite lift after $\ell$ steps is $\Lambda_{0:\ell}=\Lambda_0\circ\Lambda_1\circ\cdots\circ\Lambda_{\ell-1}$, and its image $\mathcal F_\ell=\operatorname{im}(\Lambda_{0:\ell})$ in the original spin space is the selected sector family. We call this iterative contraction and lifting procedure parity renormalization.

For any assignment $s_\ell$ of $H^{(\ell)}$, applying Eq.~\eqref{eq:parity-embedding} to each factor of $\Lambda_{0:\ell}$ in turn shows that the lifted assignment $z=\Lambda_{0:\ell}(s_\ell)$ satisfies, for every $u\in V^{(0)}$,
\begin{equation}
z_u=\alpha_u\,s_A,\qquad \alpha_u\in\{\pm1\},
\label{eq:cluster-lift}
\end{equation}
where $A\in V^{(\ell)}$ and the sign $\alpha_u$ are determined by the contractions involving $u$. For each $A\in V^{(\ell)}$, we write $V_A^{(0)}$ for the cluster of all $u\in V^{(0)}$ associated with $A$ in Eq.~\eqref{eq:cluster-lift}, and these clusters partition $V^{(0)}$. At $\ell=0$ the map $\Lambda_{0:0}$ is the identity, so $V_u^{(0)}=\{u\}$ and $\alpha_u=1$. Substituting Eq.~\eqref{eq:cluster-lift} into $H^{(0)}$ turns each coupling within a cluster into a constant and each coupling between clusters into a two-body product, and the effective coupling between $A_1$ and $A_2$ is
\begin{equation}
J'_{A_1A_2}
=
\sum_{u\in V_{A_1}^{(0)}}
\sum_{v\in V_{A_2}^{(0)}}
\alpha_u\alpha_v\,J_{uv},
\qquad A_1\ne A_2.
\label{eq:block-coupling}
\end{equation}

\begin{lemma}[Iterated lift identity]
\label{lem:iterated-lift}
For any $L\ge1$ and every assignment $s_L$ of $H^{(L)}$,
\begin{equation}
H^{(0)}(\Lambda_{0:L}(s_L))
=
H^{(L)}(s_L)
+
\sum_{\ell=0}^{L-1}\kappa_\ell.
\label{eq:iterated-lift}
\end{equation}
Lifting an exact minimizer of $H^{(L)}$ through $\Lambda_{0:L}$ gives an exact minimizer of $H^{(0)}$ over $\mathcal F_L$.
\end{lemma}

\begin{proof}
By Eq.~\eqref{eq:level-recursion}, $H^{(\ell)}\circ\Lambda_\ell=H^{(\ell+1)}+\kappa_\ell$. Applying this relation successively for $\ell=0,\ldots,L-1$ to $H^{(0)}\circ\Lambda_{0:L}$ gives Eq.~\eqref{eq:iterated-lift}. Every assignment in $\mathcal F_L$ has the form $\Lambda_{0:L}(s_L)$ for some assignment $s_L$ of $H^{(L)}$. Because $\sum_{\ell=0}^{L-1}\kappa_\ell$ does not depend on $s_L$, lifting a minimizer of $H^{(L)}$ through $\Lambda_{0:L}$ attains the minimum of $H^{(0)}$ over $\mathcal F_L$.
\end{proof}

\subsection{Sector probabilities from QAOA parity statistics}
\label{sec:qaoa-parity-statistics}

For an edge $(i,j)$ in $G^{(\ell)}$, let $c_{ij}=\langle Z_iZ_j\rangle$ denote the $ZZ$ correlation in a QAOA output state. This state may be prepared for $H^{(\ell)}$ itself or for its restriction to a subgraph of $G^{(\ell)}$ containing $(i,j)$. The computational-basis distribution of this state gives the sector probability $q_{ij}(\eta)=\Pr(s_i s_j=\eta)=(1+\eta c_{ij})/2$. For $c_{ij}\ne0$, the sign of $c_{ij}$ identifies the more probable sector, and $|c_{ij}|$ is the gap between the two sector probabilities.

For $u\in V_i^{(0)}$ and $v\in V_j^{(0)}$, Eq.~\eqref{eq:cluster-lift} gives $z_uz_v=\alpha_u\alpha_v s_is_j$ under the lift $\Lambda_{0:\ell}$. Taking expectations over the QAOA output distribution gives $\mathbb E[z_uz_v]=\alpha_u\alpha_v\,\mathbb E[s_is_j]=\alpha_u\alpha_v c_{ij}$.

Parity sector contraction uses only the computational-basis samples already taken for energy estimation, whereas coarse graining based on reduced density matrices requires additional measurement bases. Because $Z_iZ_j$ is diagonal in the computational basis, every sample from the QAOA output state at a given parameter setting yields $s_i s_j$ for all edges represented in the circuit. Thus, the same sample batch used for energy estimation also supplies estimates of the sector probabilities, without measurement-basis rotations~\cite{crawford2021pauli} or quantum-state tomography~\cite{james2001measurement}. By contrast, coarse graining based on reduced density matrices requires measurements in multiple Pauli bases to reconstruct the two-qubit reduced density matrix from up to $15$ nontrivial two-qubit Pauli-product expectation values before selecting a two-dimensional retained subspace~\cite{zou2025mqaoa}.

\subsection{Extension of the effective range}
\label{sec:effective-range}

By the identity $\mathbb E[z_uz_v]=\alpha_u\alpha_v c_{ij}$, all pairwise correlations between $V_i^{(0)}$ and $V_j^{(0)}$ are fixed by the effective correlation $c_{ij}$ up to signs determined by the earlier contractions. Therefore, repeated parity sector contraction maps local parity information on the effective graph to information that spans a longer range after lifting. A depth-$p$ QAOA circuit on $G^{(\ell)}$ has a finite light cone~\cite{hastings2019local}. For effective vertices $A_1,A_2\in V^{(\ell)}$, the statistic $\langle Z_{A_1}Z_{A_2}\rangle$ depends only on the neighborhood of radius $p$ around $A_1$ and $A_2$ in $G^{(\ell)}$. Lifting this neighborhood to the original variables gives the region
\begin{equation}
\mathcal R_p^{(\ell)}(A_1,A_2)
=
\bigcup_{A_3\in N_{G^{(\ell)}}^p(\{A_1,A_2\})}
V_{A_3}^{(0)}.
\label{eq:effective-range}
\end{equation}
Because each effective vertex in this neighborhood represents an entire cluster of original variables, the lifted region can include original variables outside the depth-$p$ light cone in $G^{(0)}$, with the growth under repeated contraction quantified in Appendix~\ref{app:effective-range-coverage}.

The contractions that produce this extension also fix parity sectors, so the selected sector family $\mathcal F_L$ is a subset of $\{\pm1\}^{V^{(0)}}$. Therefore, the minimum of $H^{(0)}$ over $\mathcal F_L$ is at least its global minimum, with equality exactly when $\mathcal F_L$ contains a global minimizer.

\section{Hierarchy with risk budgets and lower bounds}
\label{sec:risk-hierarchy}

Parity sector contraction is exact within a fixed sector, whereas the selected sector family need not contain a global minimizer. We first bound the loss caused by this restriction for a single edge and a matching, and use the repair risks and sector probabilities to state the conditions under which a matching and its parities form an admissible selection. We then derive a deterministic lower bound on the cut contribution from edges between blocks and combine it with the cut fraction inside each block to obtain a merge lower bound at a hierarchy level. Finally, we accumulate the repair risks across contraction steps and combine them with the merge lower bound to obtain a global cut lower bound that does not require \(B\le Q\).

\subsection{Repair risk for parity sector choices}
\label{sec:risk}

For a zero-field signed Ising objective \(H(z)\), consider an edge \((i,j)\) and define the minimum external sensitivity of its endpoints by
\begin{equation}
S_{ij}^{\min}
=
\min\!\left\{
\sum_{k\ne i,j}|J_{ik}|,
\sum_{k\ne i,j}|J_{jk}|
\right\}.
\label{eq:external-sensitivities}
\end{equation}
Flipping an endpoint \(u\in\{i,j\}\) reverses the sign of every term \(J_{uk}z_uz_k\) with \(k\ne i,j\), so the absolute change in the total contribution of these terms is at most \(2\sum_{k\ne i,j}|J_{uk}|\). Therefore, if \(u\) attains \(S_{ij}^{\min}\), flipping \(u\) changes this part of the energy by at most \(2S_{ij}^{\min}\) in absolute value.

For a candidate parity \(\eta\), let \(z\) be any assignment satisfying \(z_i z_j=-\eta\), and let \(z'\) be obtained from \(z\) by flipping an endpoint \(u\) that attains \(S_{ij}^{\min}\). Then \(z'_iz'_j=\eta\), the term on edge \((i,j)\) changes by \(2J_{ij}\eta\), and the energy change satisfies
\begin{equation}
H(z')-H(z)
\le
2\bigl(S_{ij}^{\min}+J_{ij}\eta\bigr).
\label{eq:repair-step}
\end{equation}
We call such a flip a repair of edge \((i,j)\) to sector \(\eta\) and define the corresponding repair risk by
\begin{equation}
R_{ij}(\eta)
=
2\bigl[S_{ij}^{\min}+J_{ij}\eta\bigr]_+,
\qquad
[x]_+=\max\{x,0\}.
\label{eq:repair-risk}
\end{equation}

\begin{theorem}[Repair bound for an arbitrary parity]
\label{thm:repair}
Let \(m\) be the minimum value of \(H\), and \(m_{ij,\eta}\) its minimum value under the constraint \(z_i z_j=\eta\). Then \(m_{ij,\eta}\) satisfies
\begin{equation}
0
\le
m_{ij,\eta}-m
\le
R_{ij}(\eta).
\label{eq:repair-bound}
\end{equation}
If \(S_{ij}^{\min}+J_{ij}\eta<0\), every global minimizer satisfies \(z_i z_j=\eta\).
\end{theorem}

\begin{proof}
Because the constraint \(z_i z_j=\eta\) restricts the feasible set, \(m_{ij,\eta}\ge m\). If \(S_{ij}^{\min}+J_{ij}\eta<0\), Eq.~\eqref{eq:repair-step} implies that, for any violating global minimizer, repairing edge \((i,j)\) to sector \(\eta\) would produce an assignment of strictly lower energy, contradicting global optimality. Therefore, every global minimizer satisfies \(z_i z_j=\eta\), and \(m_{ij,\eta}=m\). It remains to consider \(S_{ij}^{\min}+J_{ij}\eta\ge0\). Let \(z^*\) be a global minimizer. If \(z_i^*z_j^*=\eta\), then \(m_{ij,\eta}=m\). Otherwise, \(z_i^*z_j^*=-\eta\), and the repair of edge \((i,j)\) to sector \(\eta\) produces an assignment \(\tilde z\) with \(\tilde z_i\tilde z_j=\eta\). By Eq.~\eqref{eq:repair-step} and Eq.~\eqref{eq:repair-risk}, \(H(\tilde z)-H(z^*)\le R_{ij}(\eta)\). Therefore, \(m_{ij,\eta}\le H(\tilde z)\le m+R_{ij}(\eta)\).
\end{proof}

\begin{corollary}[Criterion for strong edges]
\label{cor:strong-edge}
For \(J_{ij}\ne0\), let \(\eta_E=-\operatorname{sgn}(J_{ij})\), the parity that minimizes the term on edge \((i,j)\). Then \(J_{ij}\eta_E=-|J_{ij}|\), and the repair risk in sector \(\eta_E\) is
\begin{equation}
R_{ij}(\eta_E)
=
2\bigl[S_{ij}^{\min}-|J_{ij}|\bigr]_+.
\end{equation}
We call edge \((i,j)\) strong when \(|J_{ij}|\ge S_{ij}^{\min}\). In this case, contracting edge \((i,j)\) in sector \(\eta_E\) has zero repair risk. If \(|J_{ij}|>S_{ij}^{\min}\), every global minimizer satisfies \(z_i z_j=\eta_E\).
\end{corollary}

In pseudo-Boolean optimization, the persistency property of roof duality can also determine the values of some variables at an optimum~\cite{hammer1984roof}.

\begin{corollary}[Risk for a matching]
\label{cor:matching-risk}
For a matching \(M\) with selected parities \(\eta_{ij}\), let \(m_M\) be the minimum value of \(H\) under the constraints \(z_i z_j=\eta_{ij}\) for all \((i,j)\in M\). Then \(m_M\) satisfies
\begin{equation}
m_M-m
\le
\sum_{(i,j)\in M} R_{ij}(\eta_{ij}).
\label{eq:matching-risk}
\end{equation}
\end{corollary}

\begin{proof}
Let \(z^*\) be a global minimizer. Starting from \(z^*\), repair in turn every edge of \(M\) that violates its selected parity in the current assignment. Because \(M\) is a matching, repairing one edge does not change the current parity of any other edge of \(M\), so each edge of \(M\) is repaired at most once. Whenever an edge of \(M\) still violates its selected parity, choose such an edge \((i,j)\), for which the current assignment \(z\) satisfies \(z_i z_j=-\eta_{ij}\). By Eq.~\eqref{eq:repair-step} and Eq.~\eqref{eq:repair-risk}, repairing edge \((i,j)\) to sector \(\eta_{ij}\) produces an assignment \(z'\) with \(H(z')-H(z)\le R_{ij}(\eta_{ij})\). Summing these inequalities over the repaired edges and using \(R_{ij}(\eta_{ij})\ge0\) for the unrepaired edges of \(M\) shows that the final assignment \(\tilde z\) satisfies
\begin{equation}
H(\tilde z)-m
\le
\sum_{(i,j)\in M} R_{ij}(\eta_{ij}).
\label{eq:matching-repair-sum}
\end{equation}
The procedure ends when no edge of \(M\) violates its selected parity, so \(\tilde z\) satisfies every constraint of \(M\). Therefore, \(m_M\le H(\tilde z)\), and Eq.~\eqref{eq:matching-risk} follows.
\end{proof}

\subsection{Matching with risk budgets and the parity renormalization solver}
\label{sec:local-solver}

At each contraction step, the sector probability \(q_{ij}(\eta)\) and the repair risk \(R_{ij}(\eta)\) are computed for every candidate \(((i,j),\eta)\) on the current effective graph. Given a nonnegative risk budget \(T_R\) and a nonnegative threshold \(T_q\), an admissible selection consists of a matching \(M\) and parities \(\eta_{ij}\) satisfying
\begin{equation}
\sum_{(i,j)\in M} R_{ij}(\eta_{ij}) \le T_R,
\qquad
\sum_{(i,j)\in M}\bigl(1-q_{ij}(\eta_{ij})\bigr) \le T_q.
\label{eq:budgets}
\end{equation}
By Corollary~\ref{cor:matching-risk}, imposing the constraints associated with any admissible selection increases the minimum value of the current objective by at most \(T_R\). The matching rule used in the numerical calculations is specified in Appendix~\ref{app:sector-selection-configurations}.

If all \(q_{ij}(\eta_{ij})\) are marginals of a common QAOA output distribution on the current effective graph, the union bound gives
\begin{equation}
\Pr\!\left(
\bigcap_{(i,j)\in M}
\{s_is_j=\eta_{ij}\}
\right)
\ge
1-\sum_{(i,j)\in M}\bigl(1-q_{ij}(\eta_{ij})\bigr).
\label{eq:union-bound-probability}
\end{equation}
Together with the second constraint in Eq.~\eqref{eq:budgets}, the common output distribution assigns probability at least \(\bigl[1-T_q\bigr]_+\) to the event that all selected parity constraints hold, with finite-sample bounds given in Appendix~\ref{app:finite-sample}. When the \(q_{ij}\) are obtained from truncated local subproblems, they are marginals of the output distributions of different local problems and are not jointly evaluated under a common distribution. In that setting, the second constraint serves only as a heuristic selection condition, and \(T_q\) has no joint-probability interpretation. The deterministic conclusion drawn from the first constraint remains valid.

The matching rule and parity sector contraction are applied repeatedly until the effective problem contains at most \(Q\) variables. Correlations are estimated on QAOA subproblems containing at most \(Q\) variables. Therefore, every quantum call used for correlation estimation or for solving the final effective problem involves at most \(Q\) variables. We call the resulting solver parity-RG (PRG).

Let \(\Delta^{\rm solve}\) denote an upper bound on the amount by which the objective value of the solution returned for the final effective problem exceeds its minimum. We take \(\Delta^{\rm solve}=0\) when the final effective problem is solved exactly.

\subsection{Recursive hierarchy and merge lower bound}
\label{sec:merge-guard}

At a hierarchy level, suppose an instance is partitioned into blocks \(V_1,\ldots,V_h\) and each block is solved separately. The resulting block assignments are then combined into a global assignment by the block flip merge of Sec.~\ref{sec:merge}. Propagating lower bounds on the cut contribution inside each block to the full problem requires a deterministic lower bound on the fraction of the total between-block edge weight cut by the merge.

For a nonnegative MaxCut instance, fix the assignment \(x^{(r)}\) inside each block and let \(W_{\rm cross}\) denote the total weight of edges between blocks. The contribution to the cut from edges between blocks can then be written as
\begin{equation}
C_{\rm cross}(\boldsymbol{\sigma})
=
\sum_{r<t}\;
\sum_{\substack{u\in V_r\\v\in V_t\\(u,v)\in E}}
W_{uv}\,
\mathbf{1}\!\bigl[
\sigma_r\ne x_u^{(r)}x_v^{(t)}\,\sigma_t
\bigr].
\label{eq:cross-block-parity-cut}
\end{equation}
From Eq.~\eqref{eq:cross-block-parity-cut}, \(C_{\rm cross}\) is a nonnegative weighted sum of parity mismatch indicators in the block flip variables \(\sigma_r\), with weights summing to \(W_{\rm cross}\). For the signed merge objective, applying Eq.~\eqref{eq:signed-term-parity-indicator} to each term \(-K_{rt}\sigma_r\sigma_t\) with \(K_{rt}\ne0\) gives
\begin{equation}
-K_{rt}\,\sigma_r\sigma_t
=
-|K_{rt}|
+
2|K_{rt}|\,
\mathbf{1}\!\bigl[\sigma_r\ne\operatorname{sgn}(K_{rt})\,\sigma_t\bigr].
\label{eq:signed-to-parity-cut}
\end{equation}
Summing over all nonzero \(K_{rt}\) gives
\begin{equation}
\begin{aligned}
C_{\rm merge}^{+}(\boldsymbol{\sigma})
&=-H_{\rm merge}(\boldsymbol{\sigma})\\
&=C_0+
\sum_{\substack{r<t\\K_{rt}\ne0}}
w_{rt}\,\mathbf{1}[\sigma_r\ne\chi_{rt}\sigma_t],\\
C_0
&=-\kappa_{\rm merge}
-\sum_{\substack{r<t\\K_{rt}\ne0}}|K_{rt}|,\\
w_{rt}&=2|K_{rt}|,\qquad
\chi_{rt}=\operatorname{sgn}(K_{rt}).
\end{aligned}
\label{eq:parity-cut-representation}
\end{equation}
Therefore, maximizing \(C_{\rm merge}^{+}\) is equivalent to minimizing \(H_{\rm merge}\).

\begin{lemma}[Merge bound by conditional expectation]
\label{lem:ce}
For a merge objective written in the form~\eqref{eq:parity-cut-representation}, there exists an assignment \(\boldsymbol{\sigma}_{\rm ce}\), obtainable by the method of conditional expectation~\cite{alon2016probabilistic}, such that
\begin{equation}
C_{\rm merge}^{+}(\boldsymbol{\sigma}_{\rm ce})
\ge
C_0+\frac12
\sum_{\substack{r<t\\K_{rt}\ne0}}w_{rt}.
\label{eq:ce-guard}
\end{equation}
\end{lemma}

\begin{proof}
Because each indicator in Eq.~\eqref{eq:parity-cut-representation} has expectation \(1/2\) under independent uniform choices of the variables \(\sigma_r\), the expected value of \(C_{\rm merge}^{+}\) is \(C_0+\frac12\sum_{\substack{r<t\\K_{rt}\ne0}}w_{rt}\). The method of conditional expectation fixes the \(\sigma_r\) one at a time without decreasing the conditional expectation, yielding a deterministic assignment satisfying Eq.~\eqref{eq:ce-guard}.
\end{proof}

For a partitioned nonnegative MaxCut instance with positive total edge weight, let \(W_{\rm in}\) be the total weight of edges within blocks, and define
\begin{equation}
\varepsilon=\frac{W_{\rm cross}}{W_{\rm tot}},
\qquad
W_{\rm tot}=W_{\rm in}+W_{\rm cross}.
\end{equation}
Suppose that the cut contribution inside each block is at least \(\tau\) times its total internal edge weight and that the chosen block flips give a contribution of at least \(\zeta W_{\rm cross}\) from edges between blocks, where \(\zeta\in[0,1]\). For the resulting assignment \(\hat z\), summing the cut contributions inside the blocks and adding the contribution from edges between blocks gives \(C_G(\hat z)\ge\tau W_{\rm in}+\zeta W_{\rm cross}\). Substituting \(W_{\rm in}=(1-\varepsilon)W_{\rm tot}\) and \(W_{\rm cross}=\varepsilon W_{\rm tot}\) gives
\begin{equation}
C_G(\hat z)
\ge
\bigl[(1-\varepsilon)\tau+\varepsilon\zeta\bigr]W_{\rm tot}.
\label{eq:general-decomp}
\end{equation}

\begin{theorem}[Merge lower bound at a hierarchy level]
\label{thm:guard}
For a nonnegative MaxCut instance, suppose that the cut contribution inside each block is at least \(\tau\) times its total internal edge weight. If the block flips are chosen by the method of conditional expectation, then the global cut value satisfies
\begin{equation}
C_G(\hat z)
\ge
\left[
\tau+\left(\frac12-\tau\right)\varepsilon
\right]W_{\rm tot}.
\label{eq:guarded-bound}
\end{equation}
\end{theorem}

\begin{proof}
Applying the method of conditional expectation used in Lemma~\ref{lem:ce} to Eq.~\eqref{eq:cross-block-parity-cut} produces block flips whose contribution from edges between blocks is at least \(W_{\rm cross}/2\). Setting \(\zeta=1/2\) in Eq.~\eqref{eq:general-decomp} gives Eq.~\eqref{eq:guarded-bound}.
\end{proof}

The blockwise hypothesis of Theorem~\ref{thm:guard} concerns the cut contribution inside each block, not its size, and does not require \(B\le Q\). In the original QAOA\(^{2}\) hierarchy, each leaf block is encoded directly as a QAOA problem on its vertices, which requires \(B\le Q\). SP-HQAOA instead handles each block with the PRG solver of Sec.~\ref{sec:local-solver}, whose QAOA calls use at most \(Q\) qubits even when \(B>Q\).

For a nonnegative MaxCut instance with objective \(H\) on \(|V|>B\) variables, the recursive construction proceeds as follows:
\begin{enumerate}[label=(\roman*)]
\item Choose the block partition so that \(|V_r|\le B\) for all \(r\).
\item Solve each block with the PRG solver of Sec.~\ref{sec:local-solver}, or with any solver that provides a valid \(\tau\) for the blockwise hypothesis of Theorem~\ref{thm:guard}.
\item Build the signed merge graph on the blocks using Eq.~\eqref{eq:merge-coupling}.
\item Solve the merge problem directly if it has at most \(Q\) block flip variables. Otherwise, rewrite its objective in the form of Eq.~\eqref{eq:parity-cut-representation} and apply the same procedure recursively, with the block solver at the next level providing a valid \(\tau\) for that objective.
\item Lift the chosen block flips to obtain \(\hat z\) and evaluate \(H(\hat z)\).
\end{enumerate}
If a merge solver guarantees that the contribution from edges between blocks is at least \(\zeta W_{\rm cross}\) for some \(\zeta>1/2\), Eq.~\eqref{eq:general-decomp} applies with this value of \(\zeta\).

\subsection{From local bounds to a global bound}
\label{sec:guarantees}

Consider a MaxCut block \(G_r\) induced by the vertex set \(V_r\), and let \(E_r\) denote its internal edge set. The cut bounds below use the block Ising energy \(H_r\) and the total internal edge weight \(W_r\) given by
\begin{equation}
H_r(z)=\sum_{(i,j)\in E_r}W_{ij}z_iz_j,
\qquad
W_r=\sum_{(i,j)\in E_r}W_{ij}.
\end{equation}
The cut value of \(G_r\) and its maximum then satisfy
\begin{equation}
C_{G_r}(z)=\tfrac12(W_r-H_r(z)),
\qquad
C_{G_r}^*=\tfrac12(W_r-H_r^*),
\end{equation}
where \(H_r^*=\min_z H_r(z)\). Once the sector choices are fixed, the exact offsets \(\kappa_\ell\) in Lemma~\ref{lem:iterated-lift} are independent of the spin assignment and cancel from energy differences. PRG is applied to this block for \(L_r\) contraction steps, with matching \(M_{r,\ell}\) and effective objective \(H_r^{(\ell)}\) at step \(\ell\). The selected parity of each edge \((i,j)\in M_{r,\ell}\) has repair risk \(R_{ij}^{(\ell)}(\eta_{ij})\). The resulting cumulative risk is
\begin{equation}
\Delta_r^{\rm RG}
=
\sum_{\ell=0}^{L_r-1}
\sum_{(i,j)\in M_{r,\ell}}
R_{ij}^{(\ell)}(\eta_{ij}).
\label{eq:cumulative-risk}
\end{equation}
Each \(R_{ij}^{(\ell)}\) is computed from the couplings of the effective objective \(H_r^{(\ell)}\).

Let \(\hat s_{L_r}\) be the solution obtained for the final effective problem \(H_r^{(L_r)}\), and let \(\Delta_r^{\rm solve}\) denote the bound \(\Delta^{\rm solve}\) of Sec.~\ref{sec:local-solver} for this block. Then \(\hat s_{L_r}\) satisfies
\begin{equation}
H_r^{(L_r)}(\hat s_{L_r})-\min_s H_r^{(L_r)}(s)
\le
\Delta_r^{\rm solve}.
\label{eq:terminal-gap}
\end{equation}

\begin{theorem}[Local risk bound for PRG]
\label{thm:local-risk}
Let \(\hat x_{\rm cert}^{(r)}\) be the solution obtained by lifting \(\hat s_{L_r}\) back to the variables of block \(G_r\). The energy difference from the block optimum satisfies
\begin{equation}
H_r(\hat x_{\rm cert}^{(r)})-H_r^*
\le
\Delta_r^{\rm RG}
+
\Delta_r^{\rm solve}.
\label{eq:local-energy-risk}
\end{equation}
Equivalently, the lifted cut value satisfies
\begin{equation}
C_{G_r}(\hat x_{\rm cert}^{(r)})
\ge
C_{G_r}^*
-
\frac{
	\Delta_r^{\rm RG}
	+
	\Delta_r^{\rm solve}
}{2}.
\label{eq:local-cut-risk}
\end{equation}
\end{theorem}

\begin{proof}
Let \(\mathcal F_{r,\ell}\) denote the sector family in the variables of block \(G_r\) after \(\ell\) contraction steps, and define
\begin{equation}
m_\ell=\min_{z\in\mathcal F_{r,\ell}} H_r(z).
\end{equation}
By Proposition~\ref{prop:matched-closure}, \(H_r^{(\ell)}\) is a zero-field signed Ising objective at every step. By Lemma~\ref{lem:iterated-lift}, \(H_r^{(\ell)}\) agrees with \(H_r|_{\mathcal F_{r,\ell}}\) up to a constant that does not depend on the spin assignment. Corollary~\ref{cor:matching-risk} then gives
\begin{equation}
m_{\ell+1}
\le
m_\ell
+
\sum_{(i,j)\in M_{r,\ell}}
R_{ij}^{(\ell)}(\eta_{ij}).
\end{equation}
Summing these inequalities over \(\ell\), with \(m_0=H_r^*\), gives
\begin{equation}
m_{L_r}
\le
H_r^*
+
\Delta_r^{\rm RG}.
\label{eq:local-risk-telescope}
\end{equation}
Combining Eq.~\eqref{eq:local-risk-telescope} with Eq.~\eqref{eq:terminal-gap} gives Eq.~\eqref{eq:local-energy-risk}. Using \(C_{G_r}=\tfrac{1}{2}(W_r-H_r)\) then gives Eq.~\eqref{eq:local-cut-risk}.
\end{proof}

Let \(\hat z_{\rm cert}\) denote the global assignment obtained by combining the block solutions \(\hat x_{\rm cert}^{(r)}\) with the chosen block flips.

\begin{corollary}[Global cut lower bound]
\label{cor:global}
Consider a nonnegative MaxCut instance in which the cut lower bound in Theorem~\ref{thm:local-risk} holds for every block with \(W_r>0\), and in which the chosen block flips cut at least a fraction \(\zeta\) of the total edge weight between blocks. Then the global cut value satisfies
\begin{equation}
\begin{aligned}
	C_G(\hat z_{\rm cert})
	&\ge
	\sum_{r:W_r>0}
	\max\!\left\{
	C_{G_r}^{*}
	-\tfrac12\bigl(\Delta_r^{\rm RG}+\Delta_r^{\rm solve}\bigr),
	\,0\right\}\\
	&\quad+\zeta W_{\rm cross}.
\end{aligned}
\label{eq:global-risk-direct}
\end{equation}
If the block flips are chosen by the method of conditional expectation, Eq.~\eqref{eq:global-risk-direct} applies with \(\zeta=1/2\).
\end{corollary}

\begin{proof}
For each block with \(W_r>0\), Theorem~\ref{thm:local-risk} gives
\begin{equation}
C_{G_r}(\hat x_{\rm cert}^{(r)})
\ge
C_{G_r}^{*}
-
\tfrac12\bigl(\Delta_r^{\rm RG}+\Delta_r^{\rm solve}\bigr),
\label{eq:blockwise-local-risk}
\end{equation}
while nonnegativity of the cut value gives \(C_{G_r}(\hat x_{\rm cert}^{(r)})\ge0\). Taking the larger of these bounds for each block and adding the contribution \(\zeta W_{\rm cross}\) from edges between blocks gives Eq.~\eqref{eq:global-risk-direct}.
\end{proof}

Equation~\eqref{eq:global-risk-direct} also covers \(W_{\rm in}=0\). In this case, the blockwise sum is empty and the bound reduces to \(C_G(\hat z_{\rm cert})\ge\zeta W_{\rm cross}\). For the normalized form that follows, assume \(W_{\rm in}>0\).

For each block with \(W_r>0\), define the effective cut fraction
\begin{equation}
\tau_{r,{\rm eff}}
=
\frac{C_{G_r}^{*}-\tfrac12\bigl(\Delta_r^{\rm RG}+\Delta_r^{\rm solve}\bigr)}{W_r}.
\label{eq:tau-eff-local}
\end{equation}
The corresponding blockwise truncated average is
\begin{equation}
\bar\tau_{\rm eff}^{\rm blk+}
=
\frac{1}{W_{\rm in}}
\sum_{r:W_r>0}
W_r\,\max\{\tau_{r,{\rm eff}},0\}.
\end{equation}
With this notation, substituting \(W_{\rm in}=(1-\varepsilon)W_{\rm tot}\) and \(W_{\rm cross}=\varepsilon W_{\rm tot}\) into the bound of Corollary~\ref{cor:global} gives the normalized form
\begin{equation}
C_G(\hat z_{\rm cert})
\ge
\bigl[(1-\varepsilon)\bar\tau_{\rm eff}^{\rm blk+}
+\varepsilon\zeta\bigr]W_{\rm tot},
\label{eq:global-risk}
\end{equation}
which has the same form as Eq.~\eqref{eq:general-decomp}.

If certified lower bounds \(\underline C_{G_r}\le C_{G_r}^{*}\) are available for every block with \(W_r>0\), substituting \(\underline C_{G_r}\) for \(C_{G_r}^{*}\) in Eq.~\eqref{eq:global-risk-direct} and Eq.~\eqref{eq:tau-eff-local} gives computable forms of Eq.~\eqref{eq:global-risk-direct} and Eq.~\eqref{eq:global-risk}.

\section{Numerical validation}
\label{sec:numerics}

We first confirm the contraction identities and the repair-risk bound in the numerical calculations, then measure how the risk budget limits the number of contracted edges, compare PRG with projection-based baselines on leaf blocks, and determine whether the hierarchy output improves when the leaf block size exceeds the qubit limit. Unless stated otherwise, all QAOA simulations use depth \(p=1\). The parameters \(\beta\in[0,\pi/2]\) and \(\gamma\in[0,\pi]\) each take nine equally spaced candidate values, with at most two coordinate descent sweeps, and Appendix~\ref{app:grid-sensitivity} gives the dependence of the QAOA reference on the number of candidate values. The \(ZZ\) statistics used for sector selection are obtained by simulating QAOA problems containing at most \(Q\) variables, with the effect of finite sampling examined in Appendix~\ref{app:finite-shot-sensitivity}. The leaf block quality and hierarchy comparisons use the matching rule specified in Appendix~\ref{app:sector-selection-configurations}, without a finite risk budget \(T_R\) or threshold \(T_q\). For instances with a known exact optimum, we report \(\mathcal A_\star=C/C^\star\), where \(C^\star\) is the exact optimum. For the remaining instances, we report \(\mathcal A_{\rm ref}=C/C_{\rm ref}\), where \(C_{\rm ref}\) is the best cut value found by classical multistart local search.

\subsection{Algebraic consistency checks}
\label{sec:invariant-checks}

The contraction identities and the repair-risk bound hold for arbitrary zero-field signed Ising objectives and are used to check the effective objectives, lifted assignments, and repair risks produced by the numerical calculations. The checks use signed Erd\H{o}s--R\'enyi instances small enough to enumerate all assignments of each contracted objective and to obtain \(m\) and \(m_{ij,\eta}\) exactly by exhaustive search. Figure~\ref{fig:consistency}(a) shows that Eq.~\eqref{eq:sector-exact} holds exactly after one matched contraction and that Eq.~\eqref{eq:iterated-lift} holds exactly after two matched contractions for every enumerated assignment. In Fig.~\ref{fig:consistency}(b), the realized sector loss \(m_{ij,\eta}-m\) satisfies Eq.~\eqref{eq:repair-bound} for every tested combination of an edge and a sector.

\begin{figure}[!tbp]
\centering
\includegraphics[width=\columnwidth]{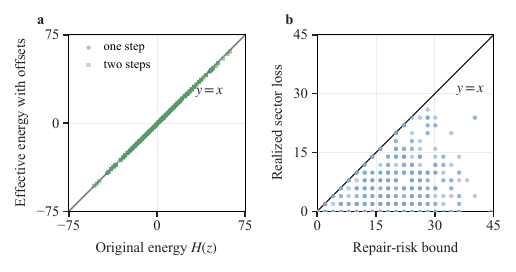}
\caption{Algebraic consistency checks. Panel (a) evaluates both sides of Eq.~\eqref{eq:sector-exact} for one matched contraction of two edges on \(20\) signed Erd\H{o}s--R\'enyi instances with \(N=7\) and edge probability \(0.5\), covering all \(32\) assignments of each contracted objective and \(640\) assignments in total. It also evaluates both sides of Eq.~\eqref{eq:iterated-lift} for two successive matched contractions of three and two edges on one signed Erd\H{o}s--R\'enyi instance with \(N=14\) and edge probability \(0.4\), covering all \(512\) assignments of the objective after the second contraction. The maximum absolute residual is \(0\) in both cases. Panel (b) compares the realized sector loss \(m_{ij,\eta}-m\) with the repair risk \(R_{ij}(\eta)\) for \(932\) combinations of an edge and a sector from \(30\) signed Erd\H{o}s--R\'enyi instances with \(N=8\) and edge probability \(0.55\). All points lie on or below the equality line.}
\label{fig:consistency}
\end{figure}

\subsection{Repair risk across graph families and risk budgets}
\label{sec:risk-families-budgets}

We quantify the tightness of the upper bound in Eq.~\eqref{eq:repair-bound} across different coupling distributions. For each tested combination of an edge and a sector, the realized sector loss is obtained by exhaustive search and compared with \(R_{ij}(\eta)\).

Four graph families supply these distributions, namely \(3\)-regular, weighted Erd\H{o}s--R\'enyi~\cite{gilbert1959random}, heavy-tailed, and signed Erd\H{o}s--R\'enyi graphs. The \(3\)-regular graphs have unit edge weights. The other three families are Erd\H{o}s--R\'enyi graphs with edge probability \(0.3\) and differ only in the weight distribution. The weighted family draws its weights uniformly from \(\{1,\ldots,5\}\). The signed family draws its weight magnitudes from the same uniform distribution and assigns independent equiprobable signs. In the heavy-tailed family, the weights are \(1+X\), where \(X\) follows a Pareto distribution with exponent \(1.5\). The vertex counts are \(N\in\{10,12,14,16\}\), with \(30\) random instances for each family and vertex count. Across all \(21{,}630\) tested combinations of an edge and a sector, the realized sector loss satisfies Eq.~\eqref{eq:repair-bound}. Among the combinations with \(R_{ij}(\eta)>0\), the upper bound is attained in every graph family, yet the median of \(\bigl(m_{ij,\eta}-m\bigr)/R_{ij}(\eta)\) does not exceed \(0.085\) in any family and vanishes in the \(3\)-regular family.

A finite risk budget \(T_R\) restricts the admissible selections, and varying it on the three families with nonnegative edge weights relates the number of contracted edges to the resulting cut quality. For each instance and each budget, we make an admissible selection, contract all selected edges simultaneously, and solve the contracted objective exactly, with budgets reported as \(T_R/(2\wtot)\). Figure~\ref{fig:risk-cert}(b) shows that more edges are contracted on average as \(T_R\) increases, yet the mean \(\mathcal{A}_\star\) is at least \(0.990\) for every family and every finite budget on these instances. The budget \(T_R=0\) admits only contractions with zero repair risk, and Appendix~\ref{app:safe-preprocess} applies this special case as a preprocessing step on larger instances.

\begin{figure}[!tbp]
\centering
\includegraphics[width=\columnwidth]{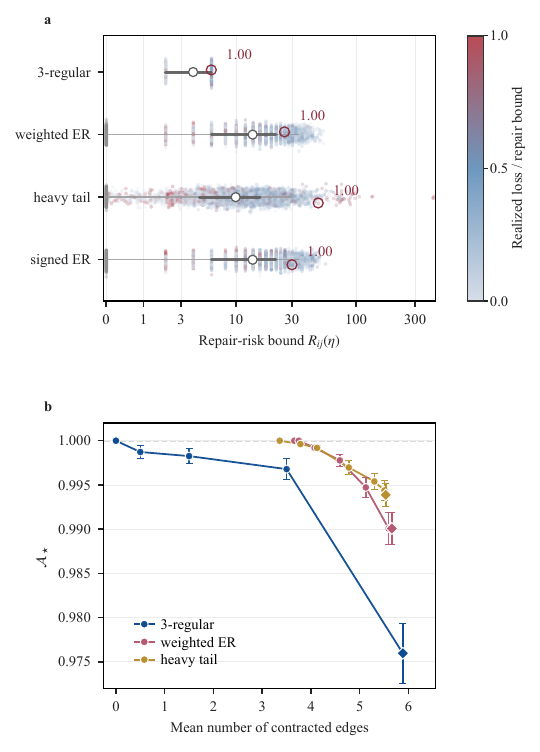}
\caption{Repair risk, solution quality, and contracted edge counts under risk budgets.
(a) Distribution of the repair risk \(R_{ij}(\eta)\) over tested combinations of an edge and a sector in four graph families, namely \(3\)-regular, weighted Erd\H{o}s--R\'enyi, heavy-tailed, and signed Erd\H{o}s--R\'enyi graphs. Horizontal positions use the transformation \(\log[1+R_{ij}(\eta)]\), with tick labels showing the corresponding original \(R_{ij}(\eta)\) values. For \(R_{ij}(\eta)>0\), point colors represent the ratio \(\bigl(m_{ij,\eta}-m\bigr)/R_{ij}(\eta)\). Gray points represent combinations satisfying \(R_{ij}(\eta)=m_{ij,\eta}-m=0\) and are excluded from the color mapping. White circles, thick horizontal bars, and thin horizontal bars indicate the median, interquartile range, and \(5\)th--\(95\)th percentile range, respectively. Red open circles and their annotations mark the largest value of this ratio in each graph family. (b) Relation between \(\mathcal{A}_\star\) and the mean number of contracted edges for the three families with nonnegative edge weights. The points at finite budgets correspond, in order, to \(T_R/(2W_{\rm tot})=0, 0.01, 0.02, 0.05, 0.10, 0.20\). Each curve connects the means for the same graph family across risk budgets. Diamonds denote contraction without an upper limit on the risk budget. Vertical error bars denote the standard error of the mean (s.e.m.) of \(\mathcal{A}_\star\).}
\label{fig:risk-cert}
\end{figure}

\subsection{Parity sector restriction and leaf block quality}
\label{sec:leaf-block-quality}

On leaf blocks, we compare PRG with two baselines built from the subspace projection used in MQAOA, using the same QAOA setting for all methods. Because this projection does not preserve zero-field closure, comparison with the less constrained baselines measures the cost in cut quality of the parity sector restriction.

We use \(20\) random \(N=40\) instances from each of the \(3\)-regular and \(4\)-regular graph families. The exact optimum \(C^\star\) is obtained by integer linear programming with zero optimality gap, following the procedure used in MQAOA~\cite{zou2025mqaoa}. The \({\rm MQ1}\) baseline performs one contraction using the MQAOA subspace projection and then solves the contracted objective. The \({\rm MQ}_{\rm cl}\) baseline applies the same projection-based contraction recursively and adds classical refinement through alternative matchings and local descent after lifting. The \({\rm PRG}_{\rm cl}\) method adds the same two refinements to PRG. The one-body fields generated by the projection are retained in both baselines. Appendix~\ref{app:sector-selection-configurations} compares alternative rules for sector selection on five graph families.

Figure~\ref{fig:leaf}(a) and Table~\ref{tab:leaf-results} compare QAOA, \({\rm MQ1}\), \({\rm MQ}_{\rm cl}\), PRG, and \({\rm PRG}_{\rm cl}\) on the same instances. The QAOA entry serves as an expectation reference, whereas the contraction methods are evaluated using the realized cut values of lifted assignments, with the corresponding QAOA sampling reference given in Appendix~\ref{app:sector-selection-configurations}. QAOA reaches about \(74\%\) of the exact optimum, consistent with the finite light cone at \(p=1\)~\cite{hastings2019local,farhi2020typical}. Without classical refinement, PRG attains \(96.5\%\)--\(97.1\%\) of the exact optimum and \({\rm MQ1}\) attains \(97.7\%\)--\(98.2\%\). With classical refinement, \({\rm PRG}_{\rm cl}\) reaches \(99.4\%\)--\(99.7\%\) and \({\rm MQ}_{\rm cl}\) reaches \(99.6\%\).

\begin{table}[!tbp]
\caption{Leaf block cut quality comparison on \(N=40\) graphs. Each entry is the mean \(\mathcal{A}_\star\) over \(20\) random instances. The five methods are QAOA, one contraction based on the MQAOA subspace projection followed by solution of the contracted objective (\({\rm MQ1}\)), the recursively and classically refined subspace-projection baseline (\({\rm MQ}_{\rm cl}\)), PRG, and PRG with the same classical refinements (\({\rm PRG}_{\rm cl}\)). The QAOA column reports the expected cut value of the optimized quantum state, \(\langle C\rangle/C^\star\). The four contraction methods report the realized cut values \(C(z)/C^\star\) of lifted assignments.}
\label{tab:leaf-results}
\begin{ruledtabular}
\begin{tabular}{c c c c c c}
	\(d\) & QAOA & \({\rm MQ1}\) & \({\rm MQ}_{\rm cl}\) & PRG & \({\rm PRG}_{\rm cl}\) \\
	\(3\) & \(0.744\) & \(0.977\) & \(0.996\) & \(0.971\) & \(0.997\) \\
	\(4\) & \(0.734\) & \(0.982\) & \(0.996\) & \(0.965\) & \(0.994\) \\
\end{tabular}
\end{ruledtabular}
\end{table}

Figure~\ref{fig:leaf}(b) reports the paired differences between PRG and \({\rm MQ1}\) and between \({\rm PRG}_{\rm cl}\) and \({\rm MQ}_{\rm cl}\). PRG has lower mean cut quality than \({\rm MQ1}\) in both graph families, with a larger gap for \(4\)-regular graphs. The differences between \({\rm PRG}_{\rm cl}\) and \({\rm MQ}_{\rm cl}\) are concentrated near zero, with means \(+0.001\) on \(3\)-regular graphs and \(-0.002\) on \(4\)-regular graphs. Each mean is smaller in magnitude than the standard error of the mean of the paired differences. Therefore, relative to the less constrained \({\rm MQ}_{\rm cl}\) baseline, the parity sector restriction in \({\rm PRG}_{\rm cl}\) has no measurable cost in mean cut quality on these instances.

\begin{figure}[t]
\centering
\includegraphics[width=\columnwidth]{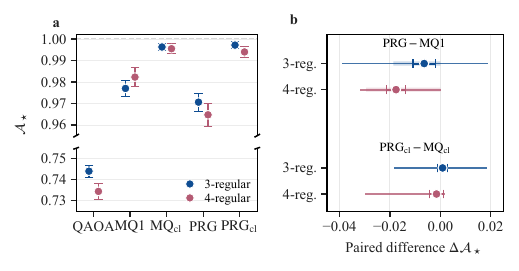}
\caption{Comparison of leaf block quality on \(N=40\) graphs. (a) Cut quality \(\mathcal{A}_\star=C/C^\star\) for five methods on \(3\)-regular and \(4\)-regular graphs. Each graph family contains \(n=20\) instances, and \(C^\star\) is the exact optimum. QAOA reports the expected cut value of the optimized quantum state, \(\langle C\rangle/C^\star\). The \({\rm MQ1}\), \({\rm MQ}_{\rm cl}\), PRG, and \({\rm PRG}_{\rm cl}\) methods report the realized cut values \(C(z)/C^\star\) of lifted assignments. Markers show the mean and error bars show the s.e.m. over instances. The dashed line marks the exact optimum, and the vertical axis is broken between the QAOA values and those of the four contraction methods. (b) Paired differences computed for each graph instance, \(\mathcal{A}_\star({\rm PRG})-\mathcal{A}_\star({\rm MQ1})\) and \(\mathcal{A}_\star({\rm PRG}_{\rm cl})-\mathcal{A}_\star({\rm MQ}_{\rm cl})\). Positive values favor PRG. Thin lines show the \(5\)th--\(95\)th percentile range and shaded bands show the interquartile range of these differences over instances, whereas markers show means with the s.e.m. of the paired differences. The dashed line marks zero difference.}
\label{fig:leaf}
\end{figure}

\subsection{Recursive hierarchy with \texorpdfstring{\(B>Q\)}{B>Q} leaf blocks}
\label{sec:large-scale-recursion}

The \(B=Q\) reference follows the divide-and-merge structure of QAOA\(^{2}\)~\cite{zhou2023qaoa2}, with each leaf block solved directly by one QAOA call. In SP-HQAOA, each merge selects the better of the recursive block flip solution and the conditional expectation assignment of Lemma~\ref{lem:ce}, whereas each merge in the reference uses only the recursive block flip solution.

We use ten weighted \(3\)-regular MaxCut instances for each of \(N=500\) and \(N=1000\), with \(Q=10\) and \(B\in\{10,20,40,80\}\). For each \(N\), the reference and SP-HQAOA use the same QAOA setting and form leaf blocks by spectral ordering based on the Fiedler vector~\cite{fiedler1973algebraic}. For each \(N\), all reported values are averaged over the same ten instances, and \(\bar\varepsilon\) denotes the mean of \(\varepsilon\) over this instance set. Figure~\ref{fig:bq-sweep} displays both \(\bar\varepsilon\) and \(\mathcal{A}_{\rm ref}\) as functions of \(B/Q\), with complete numerical values listed in Appendix~\ref{app:bq-sweep-table}.

At \(B=Q\), the reference and SP-HQAOA give comparable \(\mathcal{A}_{\rm ref}\) despite using different leaf block solvers and merge rules. Within SP-HQAOA, increasing \(B/Q\) at fixed \(Q=10\) yields clear gains. Figure~\ref{fig:bq-sweep}(a) shows that \(\bar\varepsilon\) decreases over the same range of \(B/Q\), corresponding to a larger fraction of the total edge weight inside the blocks. At \(B=80\), the hierarchy output before local search reaches \(\mathcal{A}_{\rm ref}=0.956\pm0.003\) for \(N=500\) and \(0.870\pm0.003\) for \(N=1000\), with every QAOA call using at most \(10\) qubits.

\begin{figure}[!tbp]
\centering
\includegraphics[width=\columnwidth]{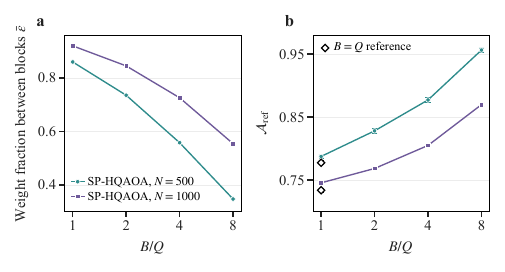}
\caption{Recursive hierarchy at fixed \(Q\) on weighted \(3\)-regular MaxCut with \(Q=10\).
(a) Mean fraction of total edge weight that lies between blocks, \(\bar\varepsilon\), as a function of \(B/Q\).
(b) \(\mathcal{A}_{\rm ref}\) before local search. In (b), diamonds mark the \(B=Q\) reference at \(B/Q=1\). Curves show SP-HQAOA for \(N=500\) and \(N=1000\). Error bars denote s.e.m. over ten instances and are smaller than the markers where not visible.}
\label{fig:bq-sweep}
\end{figure}

\section{Discussion}
\label{sec:discussion}

Through parity sector contraction, SP-HQAOA fixes selected relative parities and keeps every reduced objective in the zero-field signed Ising class, allowing the same contraction and merge rules to be applied at each recursive level. Within a single algorithm, every quantum call involves at most \(Q\) variables, and successive contractions extend the effective range of a depth-\(p\) circuit.

The separation of the maximum leaf block size \(B\) from the qubit limit \(Q\) is one of the main results of this construction. Successive contractions remove the restriction that a single quantum call places on the leaf block size, allowing leaf blocks to grow at fixed \(Q\) and to retain more edge weight within the blocks before block flip merging. Therefore, SP-HQAOA can reuse a small quantum register, and the recursive hierarchy addresses combinatorial optimization problems whose variable count far exceeds the qubit limit of a single quantum call.

Repair risk provides a quantitative constraint on successive contractions. The deterministic repair-risk bound remains valid across contraction steps. Future extensions could tighten the bound derived from a single spin flip through repairs involving multiple spins or local exact enumeration, and adaptive allocation of the risk budget could concentrate deeper contractions in blocks where they are expected to yield the largest gain. Together, these extensions could permit additional contraction steps while retaining a nontrivial cumulative lower bound.

The zero-field signed Ising class includes not only MaxCut at the top level but also the signed merge objectives that arise naturally at every recursive hierarchy level. Because any reduction that generates one-body field terms takes the intermediate effective objective outside this class, preserving the zero-field structure is what allows the same contraction rules and repair-risk bounds to be used at every level. Extending the contraction rules and repair-risk analysis to Ising Hamiltonians with one-body fields would further expand the scope of SP-HQAOA.

Sector selection requires only computational-basis samples, so on hardware it needs no measurement bases beyond those used for energy estimation. The primary direction for future work is to test parity sector contraction on quantum hardware. Zero-field closure and the repair-risk bounds remain valid whether \(c_{ij}\) is obtained from simulation or hardware measurement. The central empirical question is whether sector selection remains stable under noise. Increasing hardware scale would also relax the constraints on the hierarchy. A larger \(Q\) would reduce the number of contractions needed to compress a leaf block to a single call and make larger leaf blocks easier to accommodate within a fixed risk budget. Deeper circuits could cover larger regions of original variables in fewer contraction steps.

\begin{acknowledgments}
This work was supported by the Natural Science Foundation of Jilin Province under Grant No. 20240101013JC, and the National Natural Science Foundation of China under Grants No.12375020, No. 62201493, No. 62561052, and No. 12575010.
\end{acknowledgments}

\section*{Data and Code Availability}
The data and code that support the findings of this article are openly available~\cite{code}.

\appendix

\section{Bounds on retained probability from finite samples}
\label{app:finite-sample}

The matching rule in Sec.~\ref{sec:local-solver} relies on the sector probability \(q_{ij}(\eta)\), which must be estimated from finite samples. Reusing the same samples to choose and evaluate a matching introduces adaptive selection bias. To construct a confidence bound from finite samples without this bias, the QAOA parameters, matching, and parities are determined from one sample set, while the corresponding probabilities are estimated from an independent set of reserved samples.

For a fixed matching \(M\) whose parities \(\eta_{ij}\) are evaluated under a common QAOA output distribution on the current effective graph, define the joint sector probability as
\begin{equation}
q(M)
=
\Pr\!\left[
s_i s_j=\eta_{ij}
\ \forall\,
(i,j)\in M
\right].
\label{eq:joint-sector-probability}
\end{equation}
We call \(q(M)\) the joint retained probability of the selected parity constraints on \(M\).

\begin{proposition}[Lower bounds from reserved samples]
\label{prop:heldout}
Let \(\delta_{\rm global}\in(0,1)\) be the global failure probability, and let \(N_{\rm hold},N_{\rm cert}\ge1\). For each probability estimate, suppose that \(N_{\rm hold}\) independent reserved bitstrings are generated from the corresponding QAOA output distribution. Let \(N_{\rm cert}\) denote the total number of lower bounds to be reported for the joint probabilities \(q(M)\) and the edge probabilities \(q_{ij}(\eta_{ij})\). The corresponding QAOA output distributions, matchings, parity choices, and \(N_{\rm cert}\) are fixed before the reserved samples are used. For a matching \(M\) whose parity constraints are jointly evaluated on the same reserved samples, define
\begin{equation}
\widehat q(M)
=
\frac{1}{N_{\rm hold}}
\sum_{\nu=1}^{N_{\rm hold}}
\mathbf{1}\!\left[
s_i^{(\nu)}s_j^{(\nu)}
=
\eta_{ij}
\ \forall\,
(i,j)\in M
\right].
\label{eq:direct-heldout-estimator}
\end{equation}
With probability at least \(1-\delta_{\rm global}\), all \(N_{\rm cert}\) bounds hold simultaneously. In particular,
\begin{equation}
q(M)
\ge
\left[
\widehat q(M)
-
\sqrt{
	\frac{
		\log(N_{\rm cert}/\delta_{\rm global})
	}{
		2N_{\rm hold}
}}
\right]_+.
\label{eq:direct-heldout-bound}
\end{equation}
For each edge \((i,j)\), let \(q_{ij,{\rm lb}}\) be the lower bound obtained in the same way from the empirical frequency of \(s_i s_j=\eta_{ij}\) in the reserved samples. If the lower bound for every edge in \(M\) is counted in \(N_{\rm cert}\), then
\begin{equation}
q(M)
\ge
\left[
1-\sum_{(i,j)\in M}(1-q_{ij,{\rm lb}})
\right]_+.
\label{eq:edgewise-heldout-bound}
\end{equation}
\end{proposition}

\begin{proof}
For a fixed matching \(M\), each reserved bitstring contributes one when all parity constraints in \(M\) are satisfied and zero otherwise. Hence \(\widehat q(M)\) is the mean of \(N_{\rm hold}\) independent Bernoulli variables with expectation \(q(M)\). The same argument applies to the empirical sector probability of each edge. A one-sided Hoeffding inequality~\cite{hoeffding1963probability} gives a failure probability of at most \(\delta_{\rm global}/N_{\rm cert}\) for each lower bound. A Bonferroni correction over the \(N_{\rm cert}\) bounds, equivalently a union bound over their failure events, then shows that all bounds hold simultaneously with probability at least \(1-\delta_{\rm global}\). Substituting the individual edge bounds into Eq.~\eqref{eq:union-bound-probability} gives the second lower bound for \(q(M)\). The positive part follows from \(q(M)\ge0\).
\end{proof}

No independence among the parity constraints on different edges is assumed. Direct estimation of \(\widehat q(M)\) requires all constraints in \(M\) to be evaluated on the same reserved samples. If separate circuits reproduce the exact edge marginals of one QAOA output distribution on the current effective graph, the individual edge bounds can still be combined through Eq.~\eqref{eq:union-bound-probability}.

These bounds on retained probability complement the deterministic repair-risk bounds of Theorem~\ref{thm:repair} and Corollary~\ref{cor:matching-risk}. The repair-risk bounds control the worst-case objective loss of a sector restriction, whereas the retained probability measures the fraction of the QAOA output distribution that satisfies all selected parity constraints on \(M\).

\section{Additional numerical results}
\label{app:extended}

This appendix collects complete numerical values and supporting results for the numerical tests of Sec.~\ref{sec:numerics}.

\subsection{Comparison of sector selection rules and reference configurations}
\label{app:sector-selection-configurations}

To distinguish the respective roles of parity correlations and the repair risk constraint, we compare eight configurations comprising four sector selection rules and four reference configurations. The comparison covers five graph families, namely \(3\)-regular, \(4\)-regular, weighted Erd\H{o}s--R\'enyi, heavy-tailed, and stochastic block model graphs~\cite{holland1983stochastic}. The \(3\)-regular, weighted Erd\H{o}s--R\'enyi, and heavy-tailed families use the ensembles defined in Sec.~\ref{sec:risk-families-budgets}, while the \(4\)-regular graphs have unit edge weights. The stochastic block model divides the vertices into two equal groups. Edges within the same group and between different groups occur with probabilities \(0.5\) and \(0.1\), respectively, and all edge weights are one. The vertex counts are \(N\in\{10,12,14,16,18,20\}\), with \(30\) random instances for each family and vertex count. The exact optimum for each original instance is obtained by exhaustive search, and the effective problems after contraction are solved exactly to isolate the effect of sector selection from errors in solving the contracted problems.

For each candidate edge \((i,j)\) and sector \(\eta\), the main numerical calculations use the composite score
\begin{equation}
\begin{aligned}
{\rm score}_{ij}(\eta)
&=
|J_{ij}|
+0.25\,\lambda_{ij}(\eta)
-R_{ij}(\eta)\\
&\quad
-0.05\,D_{ij}(\eta)
+2\,\mathbf{1}\bigl[R_{ij}(\eta)\le 10^{-12}\bigr],
\end{aligned}
\label{eq:composite-score}
\end{equation}
where \(\lambda_{ij}(\eta)=\ln\!\left[(q_{ij}(\eta)+10^{-12})/(1-q_{ij}(\eta)+10^{-12})\right]\) is the regularized log odds associated with the sector probability and \(D_{ij}(\eta)=\bigl(|J_{ij}|+S_{ij}^{\min}\bigr)\bigl(1-\eta c_{ij}\bigr)\) grows with the probability of the opposite sector and with the local coupling scale. The same coefficients are used throughout the reported calculations. In the comparison with finite \(T_R\), \(T_R=0.2\,\wtot\), where \(\wtot\) is the total edge weight of the instance. No finite value is imposed on \(T_q\).

In the numerical tests, the \(ZZ\) correlations are obtained from depth-\(p=1\) statevector QAOA on local subproblems. Each subproblem contains at most \(Q\) variables and is formed from couplings incident to either endpoint of the candidate edge. When the local neighborhood is larger than \(Q\), the subproblem retains the candidate edge and the strongest adjacent couplings. At each effective graph, the same depth-\(p=1\) parameter pair is used for all such correlations.

When all selected parities are evaluated from common samples of one QAOA state on the current effective graph, both the edgewise union bound and the direct estimator of the joint sector probability described in Appendix~\ref{app:finite-sample} retain their probability interpretation. When separately sampled circuits on untruncated light cones reproduce the exact marginals of the same global QAOA state, only the edgewise construction applies. Each reported edgewise confidence bound is included in \(N_{\rm cert}\). When a light cone is truncated so that the local problem contains at most \(Q\) variables, the resulting \(q_{ij}\) are marginals of different local problems and are used only heuristically to choose the matching and parity sectors.

In the hierarchy benchmarks, the light cone is truncated only when the neighborhood of a candidate edge contains more than \(Q\) variables. The matching is built by taking admissible sectors in decreasing composite score, subject to Eq.~\eqref{eq:budgets}. In the \(N=40\) comparison of leaf block quality, effective graphs with at most \(20\) variables are evaluated as full statevector problems. The \({\rm MQ}_{\rm cl}\) and \({\rm PRG}_{\rm cl}\) methods use four matching trials, consisting of one default trial and three trials generated by perturbing the matching scores with amplitude \(0.05\). At each lifting step, the spin assignment at the preceding level is recovered by selecting, for each paired block, the computational-basis component with the largest probability in the corresponding isometry column. The single-spin local descent applied after lifting performs up to \(100\) passes.

When the selected matching is empty and the effective problem still exceeds \(Q\) variables, a greedy local search with \(20\) reproducible random restarts is used. No finite certified value of \(\Delta^{\rm solve}\) is assigned to the resulting solution. In the hierarchy benchmarks, effective problems larger than the threshold for exact enumeration specified in Appendix~\ref{app:bq-sweep-table} are treated by the same depth-\(p=1\) QAOA setting.

The comparison uses eight configurations.
\begin{enumerate}[leftmargin=1.4em]
\item The random parity reference selects each parity \(\eta\) with equal probability.
\item The QAOA sampling reference selects the best bitstring from the QAOA samples on the original graph, without contraction.
\item The single-spin elimination reference removes spins one at a time using the energy sector \(\eta_E\). The elimination generates one-body field terms and takes the resulting objective outside the zero-field class.
\item The edge energy rule selects the sector \(\eta_E\) for each matched edge.
\item The \(ZZ\) parity rule chooses \(\eta=+1\) when \(c_{ij}\ge0\) and \(\eta=-1\) when \(c_{ij}<0\).
\item The \(ZZ\) parity rule subject to the repair risk constraint uses the same parity choice and tie convention, with the matching determined by the composite score under the risk budget \(T_R\).
\item The composite parity rule subject to the repair risk constraint includes both parity sectors of each edge in the same composite score and determines the matching under the same risk budget.
\item The exact parity reference assigns all matched edges the parities realized in a single global minimizer.
\end{enumerate}
Figure~\ref{fig:sector-selection-comparison} shows that the \(ZZ\) parity rule without the repair risk constraint exhibits a low-quality tail on weighted Erd\H{o}s--R\'enyi and heavy-tailed graphs. The repair risk constraint substantially suppresses this tail. Both rules subject to the constraint remain close to the exact parity reference in cut quality, while their total repair risk is far lower than that of the reference.

\begin{figure*}[!t]
\centering
\includegraphics[width=\textwidth]{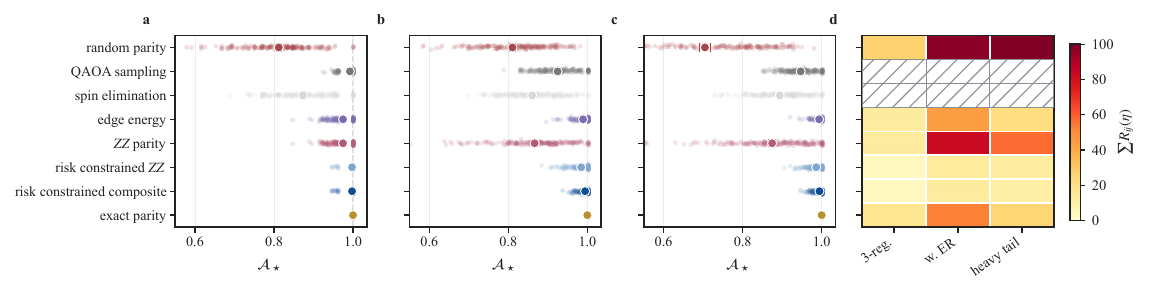}
\caption{Parity sector selection and the repair risk constraint.
(a)--(c) Performance of eight configurations on \(3\)-regular, weighted Erd\H{o}s--R\'enyi, and heavy-tailed graph families, respectively. The vertical dashed line marks \(\mathcal{A}_\star=1\). Small points show individual instances, and large markers show means with s.e.m.
(d) Mean total repair risk for each configuration across the same three graph families. Hatched cells indicate configurations for which the sum of repair risks from parity sector contraction is not applicable.}
\label{fig:sector-selection-comparison}
\end{figure*}

\subsection{Hierarchy results at fixed \texorpdfstring{\(Q\)}{Q}}
\label{app:bq-sweep-table}

Table~\ref{tab:B-sweep} gives the complete numerical values for the hierarchy benchmark at fixed \(Q\) in Sec.~\ref{sec:large-scale-recursion}. All instances are weighted \(3\)-regular MaxCut graphs with the qubit limit fixed at \(Q=10\). Spectral ordering is used to form leaf blocks of at most \(B\) vertices, and the \(B=Q\) reference and all SP-HQAOA results use the same partitioning rule.

The reference value \(C_{\rm ref}\) is the best cut found by single-spin local search from one greedy initial assignment and ten random initial assignments. The table reports hierarchy outputs before local search, and effective problems with at most six variables are solved by exact enumeration. The reported uncertainties are standard errors of the mean.

\begin{table}[!htbp]
\caption{Complete hierarchy results at fixed \(Q\). The \(B=Q\) reference solves each leaf block directly by QAOA. The reference value \(C_{\rm ref}\) is obtained by classical multistart local search, and \(\mathcal{A}_{\rm ref}=C/C_{\rm ref}\). For each \(N\), all rows use the same spectral partitioning rule and are averaged over the same ten instances.}
\label{tab:B-sweep}
\setlength{\tabcolsep}{3.2pt}
\scriptsize
\begin{ruledtabular}
\begin{tabular}{c c c c c c}
	\(N\) & \(B\) & method & blocks & \(\bar{\varepsilon}\) & \(\mathcal{A}_{\rm ref}\) \\
	\(500\) & \(10\) & \(B=Q\) reference & \(50\) & \(0.860\) & \(0.778\pm0.003\) \\
	\(500\) & \(10\) & SP-HQAOA & \(50\) & \(0.860\) & \(0.787\pm0.002\) \\
	\(500\) & \(20\) & SP-HQAOA & \(25\) & \(0.735\) & \(0.828\pm0.004\) \\
	\(500\) & \(40\) & SP-HQAOA & \(13\) & \(0.558\) & \(0.877\pm0.004\) \\
	\(500\) & \(80\) & SP-HQAOA & \(7\) & \(0.347\) & \(0.956\pm0.003\) \\
	\(1000\) & \(10\) & \(B=Q\) reference & \(100\) & \(0.920\) & \(0.734\pm0.003\) \\
	\(1000\) & \(10\) & SP-HQAOA & \(100\) & \(0.920\) & \(0.746\pm0.003\) \\
	\(1000\) & \(20\) & SP-HQAOA & \(50\) & \(0.845\) & \(0.768\pm0.003\) \\
	\(1000\) & \(40\) & SP-HQAOA & \(25\) & \(0.726\) & \(0.805\pm0.003\) \\
	\(1000\) & \(80\) & SP-HQAOA & \(13\) & \(0.554\) & \(0.870\pm0.003\) \\
\end{tabular}
\end{ruledtabular}
\end{table}

\subsection{Coverage of the effective range}
\label{app:effective-range-coverage}

The region \(\mathcal R_p^{(\ell)}(A_1,A_2)\) defined in Sec.~\ref{sec:effective-range} is obtained by lifting the radius-\(p\) neighborhood of two effective vertices to the original variables. To summarize the same lifting mechanism without specifying a distribution over edges, we instead average the fraction of original variables covered by the lifted radius-\(p\) neighborhood of a single effective vertex. The mean lifted coverage is
\begin{equation}
\rho_p^{(\ell)}
=
\mathbb E_{A_1\sim\operatorname{Unif}(V^{(\ell)})}
\left[
\frac{
\left|
\bigcup_{A_2\in N_{G^{(\ell)}}^p(A_1)}
V_{A_2}^{(0)}
\right|
}{N}
\right].
\label{eq:range-coverage}
\end{equation}
For each graph, contraction step, and circuit depth \(p\), the numerical estimate samples up to \(64\) effective vertices uniformly without replacement. All effective vertices are used when \(\lvert V^{(\ell)}\rvert\le64\). On \(N=500\) weighted \(3\)-regular graphs, contractions using a maximal matching rapidly increase the mean lifted coverage. For \(p=1\), \(\rho_1^{(\ell)}\) grows from \(0.008\) at \(\ell=0\) to \(0.803\) at \(\ell=8\). For \(p=2\), \(\rho_2^{(\ell)}\) grows from \(0.020\) at \(\ell=0\) to \(0.984\) at \(\ell=5\).

The effective graph does not become complete in this process. Instead, each effective vertex represents a growing cluster of original variables. Therefore, a radius-\(p\) neighborhood on the effective graph covers more original variables after lifting and extends the effective range described in Sec.~\ref{sec:effective-range}.

\begin{figure}[!tbp]
\centering
\includegraphics[width=\columnwidth]{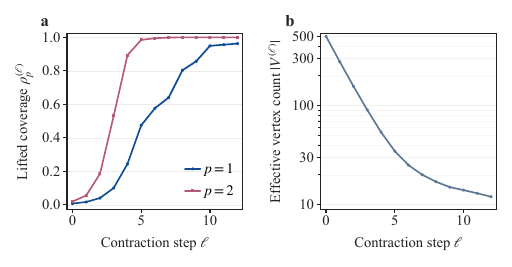}
\caption{Lifted coverage under repeated contraction of maximal matchings. Curves and shaded bands denote the mean and s.e.m. over three graph instances, respectively.
(a) Lifted coverage \(\rho_p^{(\ell)}\) for \(p=1\) and \(p=2\) on three \(N=500\) weighted \(3\)-regular graph instances. The vertical axis reports the fraction of original variables covered after a radius-\(p\) neighborhood on the effective graph is lifted back to the original variables, and faint points correspond to individual graph instances.
(b) Effective vertex count \(|V^{(\ell)}|\) across contraction steps.}
\label{fig:effective-range-validation}
\end{figure}

\subsection{Contractions with zero repair risk}
\label{app:safe-preprocess}

Corollary~\ref{cor:strong-edge} identifies edge sectors with zero repair risk. We use these sectors in a preprocessing step that contracts strong edges before the ordinary PRG step. We measure the fraction of vertices removed and compare the cut values obtained by the same solver with and without this preprocessing step.

The test uses graph sizes \(N\in\{60,100,200\}\) and five graph families, namely \(3\)-regular, \(4\)-regular, weighted Erd\H{o}s--R\'enyi, heavy-tailed, and stochastic block model graphs, with \(20\) random instances for each family and size. The cut changes are evaluated at \(Q\in\{6,8,10\}\). The vertex reduction fraction is defined as \(1-N_{\rm after}/N\) and is independent of \(Q\). We compare two variants. The first applies the contraction rule only once before solving, whereas the second reapplies it at each recursive level after contractions create new effective edges.

Figure~\ref{fig:safe-preprocess}(a) shows that the vertex reduction fraction is largest on heavy-tailed graphs, intermediate on weighted Erd\H{o}s--R\'enyi and stochastic block model graphs, and essentially zero on random regular graphs. Figure~\ref{fig:safe-preprocess}(b) compares the normalized cut change \(\Delta C/\wtot\) relative to the same procedure without this preprocessing step. Zero repair risk preserves the exact optimum value, but the realized cut value returned by an approximate solver may still increase or decrease after contraction.

\begin{figure}[!tbp]
\centering
\includegraphics[width=\columnwidth]{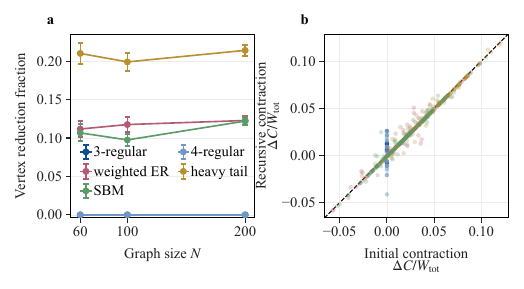}
\caption{Contractions with zero repair risk. We compare vertex reduction and normalized cut changes produced by contractions of strong edges with zero repair risk across five graph families, namely \(3\)-regular, \(4\)-regular, weighted Erd\H{o}s--R\'enyi, heavy-tailed, and stochastic block model (SBM) graphs. Panel (a) shows the vertex reduction fraction \(1-N_{\rm after}/N\) across graph sizes. Error bars show s.e.m. over \(20\) independent graph instances. Panel (b) compares the normalized cut change \(\Delta C/\wtot\) when the contraction rule is applied only once before solving or reapplied recursively. The dashed diagonal marks equal changes in the two cases. Each point represents one graph instance at one value of \(Q\).}
\label{fig:safe-preprocess}
\end{figure}

\subsection{Dependence of the QAOA reference on the number of parameter values}
\label{app:grid-sensitivity}

The main QAOA calculations optimize the QAOA parameters \(\beta\) and \(\gamma\) over nine candidate values each. To examine how the QAOA reference depends on this choice, Fig.~\ref{fig:grid-sensitivity} repeats the calculation with \(9\), \(13\), and \(17\) candidate values for each parameter on eight \(N=40\) instances from each random regular family.

On both graph families, increasing the number of candidate values per parameter from \(9\) to \(13\) or \(17\) does not close the gap between the QAOA reference and the contraction configurations in Sec.~\ref{sec:leaf-block-quality}.

\begin{figure}[!tbp]
\centering
\includegraphics[width=0.85\columnwidth]{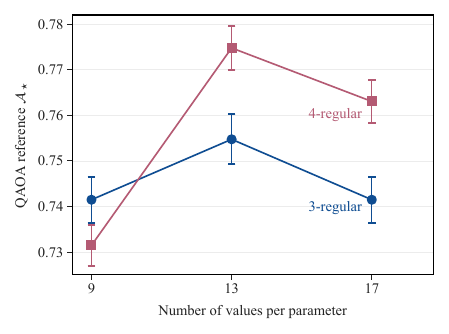}
\caption{QAOA reference for different numbers of parameter values. For QAOA at depth \(p=1\) on \(N=40\) \(3\)-regular and \(4\)-regular leaf blocks, we use \(9\), \(13\), and \(17\) candidate values for each of the parameters \(\beta\) and \(\gamma\). The vertical axis shows \(\mathcal{A}_\star=\langle C\rangle/C^\star\), where \(C^\star\) is the exact optimum. Markers and error bars show the mean and s.e.m. over eight instances from each graph family.}
\label{fig:grid-sensitivity}
\end{figure}

\subsection{Sensitivity of sector selection to finite sampling}
\label{app:finite-shot-sensitivity}

We perform sector selection on \(N=16\) instances from the \(3\)-regular, weighted Erd\H{o}s--R\'enyi, and heavy-tailed families using \(ZZ\) correlations estimated from \(128\), \(512\), and \(2048\) computational-basis samples, and compare the resulting cut values with those obtained from exact \(ZZ\) correlations. Each family contains \(20\) random instances.

Figure~\ref{fig:finite-shot} reports the cut quality and its difference from the result obtained with exact \(ZZ\) correlations. Averaged over all \(60\) instances, the mean \(\mathcal{A}_\star\) is \(0.988\), \(0.990\), and \(0.992\) for \(128\), \(512\), and \(2048\) samples, respectively, compared with \(0.990\) for exact \(ZZ\) correlations. Even with \(128\) samples, the mean cut quality differs from the exact-correlation result by only \(0.002\).

\begin{figure}[!tbp]
\centering
\includegraphics[width=\columnwidth]{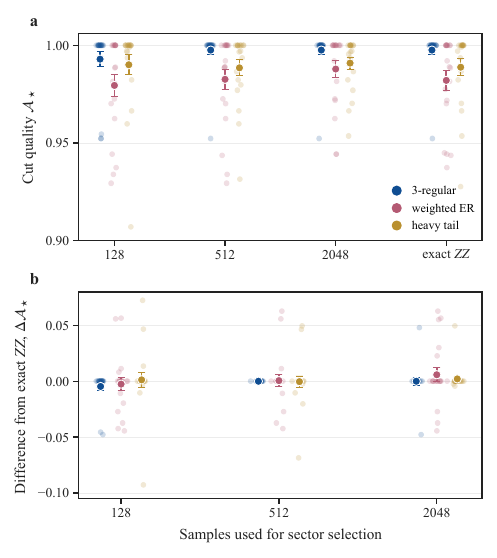}
\caption{Sensitivity of sector selection to finite sampling.
(a) Cut quality \(\mathcal{A}_\star\) obtained using \(128\), \(512\), and \(2048\) samples for sector selection and using exact \(ZZ\) correlations.
(b) Difference from the result obtained with exact \(ZZ\) correlations, defined as \(\Delta\mathcal{A}_\star=\mathcal{A}_\star(\mathrm{sampled}\,ZZ)-\mathcal{A}_\star(\mathrm{exact}\,ZZ)\). Small points show individual graph instances, whereas large markers and error bars show the mean and s.e.m. over \(20\) instances in each family.}
\label{fig:finite-shot}
\end{figure}

\makeatletter
\def\bibsection{%
\par
\begingroup
\baselineskip26\p@
\bib@device{\columnwidth}{72\p@}%
\endgroup
\nobreak\@nobreaktrue
\addvspace{19\p@}%
}
\makeatother

\bibliography{}

\begin{thebibliography}{40}

\bibitem{barahona1982complexity}
F. Barahona,
``On the computational complexity of Ising spin glass models,''
\emph{J. Phys. A: Math. Gen.} \textbf{15}, 3241--3253 (1982).

\bibitem{lucas2014ising}
A. Lucas,
``Ising formulations of many NP problems,''
\emph{Front. Phys.} \textbf{2}, 5 (2014).

\bibitem{mohseni2022ising}
N. Mohseni, P. L. McMahon, and T. Byrnes,
``Ising machines as hardware solvers of combinatorial optimization problems,''
\emph{Nat. Rev. Phys.} \textbf{4}, 363--379 (2022).

\bibitem{preskill2018nisq}
J. Preskill,
``Quantum computing in the NISQ era and beyond,''
\emph{Quantum} \textbf{2}, 79 (2018).

\bibitem{bharti2022nisq}
K. Bharti, A. Cervera-Lierta, T. H. Kyaw, \emph{et al.},
``Noisy intermediate-scale quantum algorithms,''
\emph{Rev. Mod. Phys.} \textbf{94}, 015004 (2022).

\bibitem{farhi2014qaoa}
E. Farhi, J. Goldstone, and S. Gutmann,
``A quantum approximate optimization algorithm,''
arXiv:1411.4028 (2014).

\bibitem{zhou2020qaoa}
L. Zhou, S.-T. Wang, S. Choi, H. Pichler, and M. D. Lukin,
``Quantum approximate optimization algorithm: Performance, mechanism, and implementation on near-term devices,''
\emph{Phys. Rev. X} \textbf{10}, 021067 (2020).

\bibitem{blekos2024review}
K. Blekos, D. Brand, A. Ceschini, C.-H. Chou, R.-H. Li, K. Pandya, and A. Summer,
``A review on quantum approximate optimization algorithm and its variants,''
\emph{Phys. Rep.} \textbf{1068}, 1--66 (2024).

\bibitem{goemans1995improved}
M. X. Goemans and D. P. Williamson,
``Improved approximation algorithms for maximum cut and satisfiability problems using semidefinite programming,''
\emph{J. ACM} \textbf{42}, 1115--1145 (1995).

\bibitem{wang2018qaoa}
Z. Wang, S. Hadfield, Z. Jiang, and E. G. Rieffel,
``Quantum approximate optimization algorithm for MaxCut: A fermionic view,''
\emph{Phys. Rev. A} \textbf{97}, 022304 (2018).

\bibitem{wurtz2021maxcut}
J. Wurtz and P. Love,
``MaxCut quantum approximate optimization algorithm performance guarantees for $p>1$,''
\emph{Phys. Rev. A} \textbf{103}, 042612 (2021).

\bibitem{farhi2022sk}
E. Farhi, J. Goldstone, S. Gutmann, and L. Zhou,
``The quantum approximate optimization algorithm and the Sherrington-Kirkpatrick model at infinite size,''
\emph{Quantum} \textbf{6}, 759 (2022).

\bibitem{boros2002pseudo}
E. Boros and P. L. Hammer,
``Pseudo-Boolean optimization,''
\emph{Discrete Appl. Math.} \textbf{123}, 155--225 (2002).

\bibitem{ebadi2022mis}
S. Ebadi, A. Keesling, M. Cain, \emph{et al.},
``Quantum optimization of maximum independent set using Rydberg atom arrays,''
\emph{Science} \textbf{376}, 1209--1215 (2022).

\bibitem{kim2022rydberg}
M. Kim, K. Kim, J. Hwang, E.-G. Moon, and J. Ahn,
``Rydberg quantum wires for maximum independent set problems,''
\emph{Nat. Phys.} \textbf{18}, 755--759 (2022).

\bibitem{bravyi2022coloring}
S. Bravyi, A. Kliesch, R. Koenig, and E. Tang,
``Hybrid quantum-classical algorithms for approximate graph coloring,''
\emph{Quantum} \textbf{6}, 678 (2022).

\bibitem{brandao2022sdp}
F. G. S. L. Brand\~ao, R. Kueng, and D. Stilck Fran\c{c}a,
``Faster quantum and classical SDP approximations for quadratic binary optimization,''
\emph{Quantum} \textbf{6}, 625 (2022).

\bibitem{boulebnane2024sat}
S. Boulebnane and A. Montanaro,
``Solving Boolean satisfiability problems with the quantum approximate optimization algorithm,''
\emph{PRX Quantum} \textbf{5}, 030348 (2024).

\bibitem{shaydulin2024scaling}
R. Shaydulin, C. Li, S. Chakrabarti, \emph{et al.},
``Evidence of scaling advantage for the quantum approximate optimization algorithm on a classically intractable problem,''
\emph{Sci. Adv.} \textbf{10}, eadm6761 (2024).

\bibitem{harrigan2021qaoa}
M. P. Harrigan, K. J. Sung, M. Neeley, \emph{et al.},
``Quantum approximate optimization of non-planar graph problems on a planar superconducting processor,''
\emph{Nat. Phys.} \textbf{17}, 332--336 (2021).

\bibitem{pagano2020qaoa}
G. Pagano, A. Bapat, P. Becker, \emph{et al.},
``Quantum approximate optimization of the long-range Ising model with a trapped-ion quantum simulator,''
\emph{Proc. Natl. Acad. Sci. U.S.A.} \textbf{117}, 25396--25401 (2020).

\bibitem{guerreschi2019qaoa}
G. G. Guerreschi and A. Y. Matsuura,
``QAOA for Max-Cut requires hundreds of qubits for quantum speed-up,''
\emph{Sci. Rep.} \textbf{9}, 6903 (2019).

\bibitem{farhi2020typical}
E. Farhi, D. Gamarnik, and S. Gutmann,
``The quantum approximate optimization algorithm needs to see the whole graph: A typical case,''
arXiv:2004.09002 (2020).

\bibitem{farhi2020worstcase}
E. Farhi, D. Gamarnik, and S. Gutmann,
``The quantum approximate optimization algorithm needs to see the whole graph: Worst case examples,''
arXiv:2005.08747 (2020).

\bibitem{lieb1972finite}
E. H. Lieb and D. W. Robinson,
``The finite group velocity of quantum spin systems,''
\emph{Commun. Math. Phys.} \textbf{28}, 251--257 (1972).

\bibitem{bravyi2006lieb}
S. Bravyi, M. B. Hastings, and F. Verstraete,
``Lieb-Robinson bounds and the generation of correlations and topological quantum order,''
\emph{Phys. Rev. Lett.} \textbf{97}, 050401 (2006).

\bibitem{hastings2019local}
M. B. Hastings,
``Classical and quantum bounded depth approximation algorithms,''
arXiv:1905.07047 (2019).

\bibitem{akshay2020reachability}
V. Akshay, H. Philathong, M. E. S. Morales, and J. D. Biamonte,
``Reachability deficits in quantum approximate optimization,''
\emph{Phys. Rev. Lett.} \textbf{124}, 090504 (2020).

\bibitem{zhou2023qaoa2}
Z. Zhou, Y. Du, X. Tian, and D. Tao,
``QAOA-in-QAOA: Solving large-scale MaxCut problems on small quantum machines,''
\emph{Phys. Rev. Applied} \textbf{19}, 024027 (2023).

\bibitem{li2021dcqaoa}
J. Li, M. Alam, and S. Ghosh,
``Large-scale quantum approximate optimization via divide-and-conquer,''
arXiv:2102.13288 (2021).

\bibitem{ponce2025graph}
M. Ponce, R. Herrman, P. C. Lotshaw, S. Powers, G. Siopsis, T. Humble, and J. Ostrowski,
``Graph decomposition techniques for solving combinatorial optimization problems with variational quantum algorithms,''
\emph{Quantum Inf. Process.} \textbf{24}, 60 (2025).

\bibitem{ushijima2021multilevel}
H. Ushijima-Mwesigwa, R. Shaydulin, C. F. A. Negre, S. M. Mniszewski, Y. Alexeev, and I. Safro,
``Multilevel combinatorial optimization across quantum architectures,''
\emph{ACM Trans. Quantum Comput.} \textbf{2}, 1--29 (2021).

\bibitem{bravyi2020obstacles}
S. Bravyi, A. Kliesch, R. Koenig, and E. Tang,
``Obstacles to variational quantum optimization from symmetry protection,''
\emph{Phys. Rev. Lett.} \textbf{125}, 260505 (2020).

\bibitem{zou2025mqaoa}
P. Zou,
``Multiscale quantum approximate optimization algorithm,''
\emph{Phys. Rev. A} \textbf{111}, 012427 (2025).

\bibitem{crawford2021pauli}
O. Crawford, B. van Straaten, D. Wang, T. Parks, E. Campbell, and S. Brierley,
``Efficient quantum measurement of Pauli operators in the presence of finite sampling error,''
\emph{Quantum} \textbf{5}, 385 (2021).

\bibitem{james2001measurement}
D. F. V. James, P. G. Kwiat, W. J. Munro, and A. G. White,
``Measurement of qubits,''
\emph{Phys. Rev. A} \textbf{64}, 052312 (2001).

\bibitem{hammer1984roof}
P. L. Hammer, P. Hansen, and B. Simeone,
``Roof duality, complementation and persistency in quadratic 0--1 optimization,''
\emph{Math. Program.} \textbf{28}, 121--155 (1984).

\bibitem{alon2016probabilistic}
N. Alon and J. H. Spencer,
\emph{The Probabilistic Method}, 4th ed.
(Wiley, Hoboken, NJ, 2016).

\bibitem{gilbert1959random}
E. N. Gilbert,
``Random graphs,''
\emph{Ann. Math. Stat.} \textbf{30}, 1141--1144 (1959).

\bibitem{fiedler1973algebraic}
M. Fiedler,
``Algebraic connectivity of graphs,''
\emph{Czechoslovak Math. J.} \textbf{23}, 298--305 (1973).

\bibitem{hoeffding1963probability}
W. Hoeffding,
``Probability inequalities for sums of bounded random variables,''
\emph{J. Am. Stat. Assoc.} \textbf{58}, 13--30 (1963).

\bibitem{holland1983stochastic}
P. W. Holland, K. B. Laskey, and S. Leinhardt,
``Stochastic blockmodels: First steps,''
\emph{Social Networks} \textbf{5}, 109--137 (1983).

\bibitem{code}
\url{https://github.com/981090208/SP-HQAOA-with-Subproblems-Beyond-the-Qubit-Limit}

\end{thebibliography}

\end{document}